\documentclass[review,hidelinks,onefignum,onetabnum]{siamart251216}

\usepackage{amsfonts}
\usepackage{array}
\usepackage{placeins}

\newsiamremark{remark}{Remark}

\headers{Differentially Private Delayed Consensus}
{M.-Y. Wang, X. Zong, J. Wang, and J.-F. Zhang}

\title{Differentially Private Consensus for Time-Delay Multi-Agent
Systems\thanks{\funding{This research was supported by the National Natural
Science Foundation of China under grants 62522319, 62433020, 62473347,
62261136550, and T2293770.}}}

\author{Ming-Yu Wang\thanks{School of Future Technology, China University of
Geosciences, Wuhan 430074, China.}
\and Xiaofeng Zong\thanks{School of Artificial Intelligence and Automation,
China University of Geosciences, Wuhan 430074, China; Hubei Key Laboratory of
Advanced Control and Intelligent Automation for Complex Systems, Wuhan
430074, China.}
\and Jimin Wang\thanks{School of Automation and Electrical Engineering,
University of Science and Technology Beijing, Beijing 100083, China; Key
Laboratory of Knowledge Automation for Industrial Processes, Ministry of
Education, Beijing 100083, China.}
\and Ji-Feng Zhang\thanks{Corresponding author. School of Automation and Electrical Engineering,
Zhongyuan University of Technology, Zhengzhou 450007, China; State Key
Laboratory of Mathematical Sciences, Academy of Mathematics and Systems
Science, Chinese Academy of Sciences, Beijing 100190, China
(\email{jif@iss.ac.cn}).}}

\ifpdf
\hypersetup{
  pdftitle={Differentially Private Consensus for Time-Delay Multi-Agent Systems},
  pdfauthor={Ming-Yu Wang, Xiaofeng Zong, Jimin Wang, and Ji-Feng Zhang}
}
\fi

\begin{document}
\nolinenumbers
\maketitle

\begin{abstract}
	We investigate differentially private consensus for discrete-time multi-agent systems with communication delays. Whereas a delay-free system requires only a single initial state for each agent, a delayed system must be initialized over an entire history segment. To protect this history, we introduce a history-level adjacency relation and design a Laplace mechanism whose noise scale is time-varying and may increase over time. The key technical contribution is a unified analysis based on the difference resolvent function. We derive decay estimates for the fundamental solutions of the associated delayed difference equations and use them to control both the disagreement dynamics and the sensitivity of the released information. Under explicit conditions on the consensus gain and noise scale, we establish mean square weak consensus, mean square strong consensus, and almost sure strong consensus. We further prove $\epsilon^\star$-differential privacy for the complete infinite-horizon observation sequence and derive a constructive parameter-selection rule for meeting a prescribed privacy budget while retaining the consensus guarantees. Numerical simulations illustrate the theoretical results and show that consensus and privacy can coexist even when the noise variance increases over time.
\end{abstract}


	
	\section{Introduction}

Multi-agent systems (MASs) have attracted sustained attention due to their wide
applications in cooperative control, sensor networks, smart grids,
distributed optimization, and autonomous vehicle coordination (\cite{CaoYuRenChen2013,OlfatiSaber2007,RenBeardAtkins2007,wang2026distributed}).
Consensus is a fundamental task in these applications. In practical
networked systems, however, communication delays commonly arise from network
congestion, limited bandwidth, and asynchronous updates (\cite{BlimanFerrariTrecate2008,li2026adaptive,OlfatiMurray2004}). Meanwhile,
the exchanged data may contain sensitive information, such as initial
positions, local observations, private preferences, or local decisions, which
should not be inferred by \emph{honest-but-curious} agents or eavesdroppers (\cite{liu2020differentially,nozari2017differentially,erfenwjm}). Therefore,
it is important to study consensus under both communication delays and
privacy-preservation requirements.

Consensus of MASs has been extensively studied for both
continuous-time and discrete-time protocols under fixed and switching
topologies, with recent extensions addressing communication failures and
manifold-valued agent dynamics
(\cite{bentaibi2026communication,JadbabaieLinMorse2003,krishna2026riemannian,Moreau2005,OlfatiSaber2007}). When communication delays are involved, the analysis
becomes more challenging, and various tools, including
Lyapunov--Krasovskii functionals, frequency-domain methods, and graph
theory, have been developed to characterize the delay effects (\cite{BlimanFerrariTrecate2008,LiuLiuXieZhang2011,OlfatiMurray2004,zong2019consensus}).
However, these works mainly focus on convergence and stability, without
quantifying the privacy leakage risks of exchanged data from the perspective
of differential privacy.

Differential privacy provides a rigorous and quantifiable framework for
protecting sensitive information. Unlike stochastic measurement noises that
arise naturally in multi-agent coordination \cite{chen2026noise}, privacy
mechanisms deliberately inject properly designed random noises into exchanged
information or local updates. Early studies on differentially private average consensus for discrete-time
systems were reported in \cite{huang2012differentially}. Subsequently, the
fundamental privacy--accuracy trade-off and the impossibility of exact average
consensus under differential privacy were revealed in
\cite{nozari2017differentially}. For continuous-time heterogeneous systems, a
differentially private output average consensus algorithm was proposed in
\cite{liu2020differentially}. Event-triggered and resilient variants were
further investigated in
\cite{fiore2019resilient,gao2019differentially}. More recently,
differentially private consensus has been extended to signed networks,
resilient bipartite consensus, dynamic topologies, and improved
privacy--accuracy trade-offs
(\cite{tian2025privacy,erfenwjm,wang2024improved,zhou2025dynamic}).

Despite these advances, most existing differentially private consensus results
are developed under delay-free communication conditions and typically regard
a single initial state as the private data. For delayed systems, however, the
evolution depends not only on $x_i(0)$ but also on the entire initial history of $x_i$
over $[-d,0]$. Therefore, protecting only the single initial
state is insufficient, and the whole delayed initial history should be
regarded as the private data to be protected. A key difficulty in the privacy analysis is to characterize how differences between adjacent private datasets affect the subsequent system evolution. Therefore, to derive an upper bound on the sensitivity, we first need to establish decay estimates for the fundamental solutions of the corresponding delayed difference equations.
Meanwhile, communication delays and privacy noises jointly affect the consensus
dynamics, making the convergence analysis more challenging. These issues
remain insufficiently explored in the existing literature. Related stochastic-system studies have analyzed convergence and stability under random sampling, quantized information, and history-dependent switching \cite{thanh2011observers,nguyen2016switching}.

Motivated by these observations, this paper studies differentially private
consensus for discrete-time multi-agent systems with communication delays.
We derive decay estimates for the fundamental solutions of delayed difference
equations and use them as a common basis for the convergence analysis and the
sensitivity analysis for differential privacy. Building on these estimates,
we establish mean square consensus, almost sure consensus, and an
infinite-horizon $\epsilon^\star$-differential privacy guarantee.

The main contributions of this work are outlined as follows.	
\begin{itemize}
	\item This paper achieves differentially private consensus for MASs subject to communication delays for the first time. Unlike most existing works protecting only a single initial state, delayed system evolution relies on the full initial history over the delay interval, rendering single-point protection inadequate.
	A novel adjacency relation for delayed initial histories is formulated to characterize the difference between two adjacent history segments.
	
	\item A key theoretical contribution is a unified framework centered on the difference resolvent function method, serving as the common foundation for both convergence and differential privacy analysis. Exponential decay estimates for fundamental solutions of delayed difference equations are derived, with which mean square and almost sure consensus are established. The proposed proof technique holds independent theoretical value, applicable to stability analysis of general additive-noise-driven delayed systems beyond the differential privacy context.
	
	\item A privacy-preserving mechanism with time-varying noise variance is
	developed. It is shown that the proposed algorithm can achieve both
	differential privacy and consensus even when the privacy noise increases
	over time. This relaxes the restrictions in
	\cite{huang2012differentially,nozari2017differentially,fiore2019resilient,
		gao2019differentially,tian2025privacy},
	where the noise variance is required to be constant or decaying. Moreover,
	the framework also covers the result in \cite{erfenwjm} as a special case
	when $d=0$.
	
	\item A constructive parameter design approach is provided to achieve a prescribed $\epsilon^\star$-differential privacy level over the infinite time horizon. Based on the explicit sensitivity bound derived from the difference resolvent function estimates, the noise scale parameter can be selected according to the given privacy budget $\epsilon^\star$, ensuring the mechanism meets the predefined privacy requirement while preserving consensus performance.
	
\end{itemize}


	Throughout this paper, 
	$\mathbb{R}^n$ and $\mathbb{R}^{m \times n}$ denote the $n$-dimensional Euclidean space and the set of $m \times n$ real matrices, respectively.
	$I_n$ represents the $n \times n$ identity matrix, and $\mathbf{1}_N$ denotes an $N$-dimensional column vector with all entries being ones.
	For a square matrix $M$, $\lambda(M)$ denotes its eigenvalues, $\rho(M)$ its spectral radius, and $\text{diag}\{m_1, \dots, m_n\}$ a diagonal matrix with diagonal elements $m_i$.
	The superscript $\top$ denotes the transpose of a real vector or matrix.
For complex vectors, the Euclidean norm is understood in the Hermitian sense,
$\|z\|^2=\sum_\ell |z_\ell|^2$.
	For a vector $x \in \mathbb{R}^n$, $\|x\|_1 = \sum_{i=1}^n |x_i|$ is the $L_1$-norm, and $\|x\|$ denotes the standard Euclidean norm ($L_2$-norm).
	$\text{Lap}(0, b)$ represents the zero-mean Laplace distribution with scale parameter $b > 0$.
	The notation $\mathbb{R}_{>0}$ denotes the set of positive real numbers.
	For any $k \le d$, an empty sum $\sum_{s=0}^{k-d-1} c(s)$ is conventionally understood as zero.
	
	\section{Preliminaries  and problem formulation}\label{sec:prelim}
	
	This section formulates the discrete-time privacy-preserving consensus problem  in the presence of communication delays and privacy noises. We first introduce some graph-theoretic notations and a basic decomposition lemma for the Laplacian matrix. Then, we present the network dynamics, definitions of consensus, and the differential privacy framework. Finally, we introduce the fundamental solutions for delayed difference equations, which will be used in the subsequent convergence and privacy analysis. The proofs of the two fundamental-solution lemmas are provided in Appendices A and B.
	
	\subsection{Graph theory}
	Let $\mathcal G=(\mathcal V,\mathcal E,A)$ be a weighted directed graph, where $\mathcal V=\{1,2,\dots,N\}$ is the node set, $\mathcal E\subseteq \mathcal V\times \mathcal V$ is the edge set, and $A=[a_{ij}]\in\mathbb R^{N\times N}$ is the adjacency matrix. An edge $(j,i)\in\mathcal E$ means that Agent $i$ can receive information from Agent $j$, that is, $a_{ij}>0$; otherwise, $a_{ij}=0$. We assume that $a_{ii}=0$ for all $i\in\mathcal V$. For Agent $i$, its neighbor set is denoted by
	$
	\mathcal N_i:=\{j\in\mathcal V:\, a_{ij}>0\},
	$
	and its in-degree is defined as
	$
	\kappa_i:=\sum_{j=1}^N a_{ij}.
	$
	Let $K:=\mathrm{diag}(\kappa_1,\dots,\kappa_N)$. The Laplacian matrix associated with $\mathcal G$ is
	$
	L:=K-A.
	$ Throughout this paper, the eigenvalues of the Laplacian matrix $L$ are
	denoted by $\lambda_i$, $i=1,\dots,N$, where $\lambda_1=0$.
	
	The following lemma summarizes several standard properties of the Laplacian matrix (see \cite{liduan2010}).
	
	\begin{lemma}\label{lem:lap_decomp}
		For the Laplacian matrix $L$ associated with $\mathcal G$, the following statements hold:
		\begin{enumerate}
			\item There exists a probability vector $\pi\in\mathbb R^N$ such that
			$
			\pi^\top L=0$ and $ \pi^\top\mathbf 1_N=1.
			$
			\item There exists a nonsingular matrix
$Q=[\mathbf 1_N,U]\in\mathbb R^{N\times N}$ with $\pi^\top U=0$,
so that the first row of $Q^{-1}$ is $\pi^\top$, and
			\begin{equation*}
				Q^{-1}LQ=
				\begin{pmatrix}
					0 & 0\\
					0 & L_e
				\end{pmatrix},
			\end{equation*}
			where $L_e\in\mathbb R^{(N-1)\times (N-1)}$.
			\item The graph $\mathcal G$ contains a spanning tree if and only if all eigenvalues of $L_e$ have positive real parts. In this case, the vector $\pi$ in item (1) is unique.
		\end{enumerate}
	\end{lemma}
	
	Lemma~\ref{lem:lap_decomp} allows us to separate the consensus mode from the disagreement modes. As will be seen later, this decomposition reduces the network dynamics to a family of scalar delayed stochastic difference equations.

	\subsection{Problem formulation}
	
	We consider a MAS with $N$ discrete-time first-order
	integrators:
	\begin{equation}\label{eq:system}
		x_i(k+1)=x_i(k)+u_i(k),\hspace{0.1cm}
		i=1,2,\dots,N,\hspace{0.1cm} k\in\mathbb Z_{\ge 0},
	\end{equation}
	where $x_i(k)\in\mathbb R^n$ and $u_i(k)\in\mathbb R^n$ denote the state
	and control input of Agent $i$, respectively.
	
	For delayed discrete-time consensus, a commonly used protocol takes the
	form (\cite{OlfatiMurray2004}):
	\begin{equation}\label{eq:nominal_protocol}
		u_i(k)
		=
		c(k)\sum_{j=1}^{N}a_{ij}
		\bigl(x_j(k-d)-x_i(k-d)\bigr),
	\end{equation}
	where $d\in\mathbb Z_{\ge0}$ is the communication delay and $c(k)\ge0$ is a
	scalar gain sequence. This protocol drives each agent by using delayed
	relative state information. However, in an open communication environment,
	the transmitted delayed states may be observed by eavesdroppers or
	honest-but-curious agents. As a result, private information contained in the
	initial states or delayed histories may be inferred from the exchanged data.
	
	To address this issue, we introduce a privacy-preserving communication
	mechanism. Instead of directly broadcasting $x_j(k-d)$, Agent $j$ releases
	a perturbed message
	\begin{equation}\label{eq:theta}
		\theta_j(k)=x_j(k-d)+w_j(k),
	\end{equation}
	where $w_j(k)\in\mathbb R^n$ is a privacy noise vector. The entries of
	$w_j(k)$ are independent Laplace random variables:
	$
	w_j^{(\ell)}(k){\sim}\mathrm{Lap}(0,b_j(k)),
	\ell=1,\dots,n,
	$
	where $b_j(k)>0$ is the scale parameter. Throughout the paper, the entire
family $\{w_j^{(\ell)}(k):j=1,\dots,N,\ \ell=1,\dots,n,\ k\ge0\}$
is mutually independent and independent of the initial history.
The sequences $c(k)$ and $b_j(k)$ are public and deterministic.
We write
\[
\bar b(k):=\max_{1\le j\le N}b_j(k),\qquad
b_{\min}(k):=\min_{1\le j\le N}b_j(k).
\]
All consensus assertions are understood for every deterministic initial
history $\{x_i(r):-d\le r\le0,\ i=1,\dots,N\}$.
Expectations and probabilities are with respect to the privacy noises;
the random consensus limit may depend on the initial history.
	
	Based on the received message, Agent $i$ computes the perturbed relative
	state information
	\begin{equation}\label{eq:z_def}
		\begin{aligned}
			z_{ji}(k)
			&=\theta_j(k)-x_i(k-d) \\
			&=x_j(k-d)-x_i(k-d)+w_j(k),
			j\in\mathcal N_i .
		\end{aligned}
	\end{equation}
	The privacy-preserving distributed protocol is then given by
	\begin{equation}\label{eq:protocol}
		u_i(k)=c(k)\sum_{j=1}^N a_{ij}z_{ji}(k),
		\quad i=1,2,\dots,N.
	\end{equation}
	
	We now state the consensus notions used in this paper.
	\begin{definition}\label{def:dt_ms}
		The system \eqref{eq:system} under the protocol \eqref{eq:protocol} is said
		to achieve \emph{mean square weak consensus} if, for any initial condition
		and any $i\neq j$,
		\[
		\lim_{k\to\infty}\mathbb E\|x_i(k)-x_j(k)\|^2=0.
		\]
		It is said to achieve \emph{mean square strong consensus} if there exists
		a random vector $x^*\in\mathbb R^n$ with $\mathbb E\|x^*\|^2<\infty$ such
		that
		\[
		\lim_{k\to\infty}\mathbb E\|x_i(k)-x^*\|^2=0,
		i=1,2,\dots,N.
		\]
	\end{definition}
	
	\begin{definition}\label{def:dt_as}
		The system \eqref{eq:system} under the protocol \eqref{eq:protocol} is said
		to achieve \emph{almost sure strong consensus} if there exists a random
		vector $x^*$ with $\mathbb P(\|x^*\|<\infty)=1$ such that
		$
		\mathbb P\!\Bigl(\lim_{k\to\infty}\|x_i(k)-x^*\|=0\Bigr)=1,
		i=1,2,\dots,N.
		$
	\end{definition}
	
	Next, we formalize the privacy requirement. For Agent $i$, define its
	delayed initial history by
	$
	H_i:=\{x_i(r)\}_{r=-d}^{0}.
	$
	The global delayed initial history is denoted by
	$
	H:=(H_1,\dots,H_N)\in\mathcal H
	$
	with
	$
	\mathcal H:=\mathcal H_1\times\cdots\times\mathcal H_N, \mathcal H_i:=(\mathbb R^n)^{d+1}.
	$
	This delayed history is regarded as the private data to be protected.
	
	For an integer horizon $T\ge0$, define
\[
\theta(k):=\operatorname{col}(\theta_1(k),\dots,\theta_N(k)),
\qquad
\Theta^T:=(\theta(0),\dots,\theta(T)).
\]
The corresponding global randomized privacy mechanism is
\[
\mathcal M^T:\mathcal H\longrightarrow(\mathbb R^{nN})^{T+1},
\qquad \mathcal M^T(H)=\Theta^T.
\]

\begin{remark}
		The delayed initial history $\mathcal H$ is regarded as the private data
		to be protected for the following reasons. In distributed networks, such
		as robotic networks and opinion dynamics, the initial states often encode
		sensitive information, including initial positions, private preferences,
		or local opinions. Moreover, due to the communication delay $d$, the early
		evolution of the system depends on the whole delayed trajectory
		$\{x_i(r)\}_{r=-d}^{0}$ rather than only on the single state $x_i(0)$.
		Hence, protecting only $x_i(0)$ is not sufficient. Protecting the entire
		delayed history is necessary to prevent adversaries or curious neighbors
		from inferring an agent's private initial information from the transmitted
		data.
	\end{remark}
	
	With the protected data specified, we introduce adjacent delayed histories
	and the corresponding differential privacy definition.
	
	\begin{definition}
		\label{def:global_adj_history}
		Given $\Delta_H>0$, two global histories
		$
		H^a=(H_1^a,\dots,H_N^a),
		H^b=(H_1^b,\dots,H_N^b)
		$
		are said to be $\Delta_H$-adjacent, denoted by $H^a\sim H^b$, if there
		exists an agent $i_0\in\{1,\dots,N\}$ such that
		\[
		H_i^a=H_i^b,\quad i\neq i_0,
		\]
		and
		\begin{equation}\label{eq:global_adj_history}
			\sum_{r=-d}^{0}
			\|x_{i_0,a}(r)-x_{i_0,b}(r)\|_1
			\le \Delta_H.
		\end{equation}
	\end{definition}
	
	\begin{definition}
		\label{def:global_history_dp}
		Given $\Delta_H>0$, the mechanism $\mathcal M^T$ is said to satisfy
		$\varepsilon(T)$-differential privacy if
		\[
		\mathbb P\{\mathcal M^T(H^a)\in\mathcal O\}
		\le
		e^{\varepsilon(T)}
		\mathbb P\{\mathcal M^T(H^b)\in\mathcal O\}
		\]
		holds for any two $\Delta_H$-adjacent histories $H^a\sim H^b$ and any
		Borel measurable set
		$
		\mathcal O\subseteq(\mathbb R^{nN})^{T+1}.
		$
\end{definition}

For the infinite time horizon, define
$\mathcal M^\infty(H)=(\theta(0),\theta(1),\dots)$ with output space
$(\mathbb R^{nN})^{\mathbb N_0}$ equipped with its product Borel
$\sigma$-algebra. It is $\epsilon$-differentially private if the same
inequality holds with $\mathcal M^T$ replaced by $\mathcal M^\infty$,
for every measurable event in that space. The protected output is the
global broadcast transcript. The extension to curious agents whose
own histories are unchanged under adjacency is discussed after
Theorem~\ref{thm:history_dp_global}.

Before presenting the main convergence and privacy results, we first develop
	some preliminary estimates for delayed difference equations. Since the state
	evolution depends on a finite segment of past states, both delayed initial
	histories and accumulated noises require careful treatment. The fundamental
	solution introduced below provides an explicit representation of the solution
	and serves as a key tool for the subsequent analysis.

	\subsection{Fundamental solutions for delayed difference equations}
	
	To facilitate the subsequent convergence and privacy analyses, we introduce a scalar delayed stochastic difference equation, which serves as the foundational building block for our theoretical framework:
	\begin{equation}\label{eq:dt_scalar}
		y(k+1)=y(k)-c(k)\lambda\,y(k-d)+c(k)\zeta(k),
	\end{equation}
	where $\lambda\in\mathbb C$ satisfies $\Re(\lambda)>0$, and $\{\zeta(k)\}$ is a zero-mean noise sequence with
$\mathbb E|\zeta(k)|^2<\infty$ at each time. Its second moment need not
be uniformly bounded. In the consensus analysis the relevant transformed
noises satisfy $\mathbb E|\zeta(k)|^2\le C_\zeta\bar b(k)^2$ and are
uncorrelated across time. The fundamental-solution estimates below are
deterministic and do not require a noise-moment assumption.
	
	The associated homogeneous equation is
	\begin{equation}\label{eq:dt_hom}
		y(k+1)=y(k)-c(k)\lambda\,y(k-d).
	\end{equation}
	For each initial index $r\in\{-d,\dots,0\}$, we define the history fundamental solution $\Gamma_h(\cdot,r)$ as
	\begin{equation}\label{eq:Gamma_history_h}
		\Gamma_h(k,r)=
		\begin{cases}
			1, & k=r,\\
			0, & k\neq r,
		\end{cases}
		\quad -d\le k\le 0,
	\end{equation}
	and extend it forward by
	\begin{equation}\label{eq:Gamma_recursion_h}
		\Gamma_h(k+1,r)=\Gamma_h(k,r)-c(k)\lambda\,\Gamma_h(k-d,r),
		\hspace{0.1cm} k\ge 0.
	\end{equation}
	Then, for any initial history $\{y(-d),\dots,y(0)\}$, the solution of \eqref{eq:dt_hom} can be represented as
	\begin{equation}\label{eq:dt_hom_rep}
		y(k)=\sum_{r=-d}^{0}\Gamma_h(k,r)y(r),\quad k\ge 0.
	\end{equation}
	
	To describe the state evolution from time $t$ to $k$, we define the two-parameter fundamental solution $\Gamma(\cdot,t)$ for each $t\ge 0$ by
	\begin{equation}\label{eq:Gamma_two_init}
		\Gamma(k,t)=0,\quad k<t,\quad \Gamma(t,t)=1,
	\end{equation}
	and
	\begin{equation}\label{eq:Gamma_two_rec}
		\Gamma(k+1,t)=\Gamma(k,t)-c(k)\lambda\,\Gamma(k-d,t),
		\quad k\ge t.
	\end{equation}
	Note that this discrete resolvent function is a scalar, time-varying form of the fundamental solution introduced in \cite{DiblikKhusainov:06}.
	Using these notations, the variation-of-constants formula for \eqref{eq:dt_scalar} is given by
	\begin{equation}\label{eq:dt_voc_two}
		y(k)=\hspace{-0.1cm}\sum_{r=-d}^{0}\Gamma_h(k,r)y(r)
		+\sum_{s=0}^{k-1}\Gamma(k,s+1)c(s)\zeta(s),
		\hspace{0.1cm} k\ge 0.
	\end{equation}
	
	The following two lemmas provide key exponential decay estimates for the fundamental solutions.
	
	\begin{lemma}
		\label{lem:dt_fundamental_exp_A}
		Suppose there exists an integer $k_0\ge 0$ such that
		\begin{equation}\label{eq:cbar_A}
			\bar c=\sup_{k\ge k_0}c(k)<\infty,
			\alpha=2\Re(\lambda)-(2d+1)|\lambda|^2\bar c>0.
		\end{equation}
		Choose $\varrho>0$ sufficiently small such that
		\begin{equation}\label{eq:h_varrho_condition}
			h(\varrho)=
			-\alpha
			+
			\varrho e^{\varrho\bar c}C_0
			+
			\varrho C_1 d\bar c^2 e^{\varrho(d+1)\bar c}<0,
		\end{equation}
		where $	C_0=2(1+|\lambda|\bar c)^2+d|\lambda|^2\bar c^2$ and $C_1=3d|\lambda|^2$.
		Then, there exists a constant $C\ge1$ such that
		\begin{equation}\label{eq:Gamma_exp_A}
			|\Gamma(k,t)|^2
			\le
			C\exp\!\left(
			-\varrho\sum_{i=t}^{k-1}c(i)
			\right),
			\qquad k\ge t\ge0.
		\end{equation}
	For $d\ge1$, let
	$\chi_d:=\max_{0\le s<d}|\lambda|^2c(s)^2
	\exp\!\bigl(\varrho\sum_{\ell=0}^{s}c(\ell)\bigr)$,
	and set $\chi_0:=0$. Then take
	$C_h:=C\max\{1,\chi_d\}\ge1$.
		Then, for every $k\ge0$ and $r\in\{-d,\ldots,0\}$,
		\begin{equation}\label{eq:Gammah_exp_A}
			|\Gamma_h(k,r)|^2
			\le
			C_h\exp\!\left(
			-\varrho\sum_{i=0}^{k-1}c(i)
			\right).
		\end{equation}
		Both $C$ and $C_h$ depend only on $(\lambda,d,\bar c,\varrho)$ and
		$\{c(s):0\le s\le k_0+d\}$.
\end{lemma}

An explicit finite-prefix construction of $C$ is given in Appendix A;
the corresponding $C_h$ is then determined by the formula in the lemma.
In particular, this is an exponential estimate in the cumulative-gain
clock $\sum c(i)$, rather than necessarily in the integer time $k$.

\begin{remark}
		The condition on $\varrho$ in Lemma~\ref{lem:dt_fundamental_exp_A}
		is always feasible. Indeed, by the  condition
		$\alpha>0$, we have
		$
		h(0)=-\alpha<0.
		$
		Since $h(\cdot)$ is continuous with respect to $\varrho$, there must exist
		a sufficiently small constant $\varrho>0$ such that
		$
		h(\varrho)<0.
		$
		Therefore, the required $\varrho$ in
		\eqref{eq:h_varrho_condition} always exists.
	\end{remark}
	
	\begin{lemma}\label{lem:dt_kernel_toeplitz}
		Let
		$
		a_{k,s}:=|\Gamma(k,s+1)|\,c(s), 0\le s\le k-1.
		$
		Under conditions $\sum_{k=0}^{\infty} c(k) = \infty$ and \eqref{eq:cbar_A}, the array $\{a_{k,s}\}$ satisfies
		$
		a_{k,s}\to 0$  for each fixed $s$, $
		\sup_{k\ge 1}\sum_{s=0}^{k-1}a_{k,s}<\infty.
		$
		Consequently, for any deterministic sequence $\{u(s)\}$ with $u(s)\to 0$, one has
		$
		\sum_{s=0}^{k-1}\Gamma(k,s+1)c(s)u(s)\to 0.
		$
	\end{lemma}
	
	With the above preliminaries in place, we are ready to study the asymptotic behavior of the closed-loop system and its differential privacy properties. The next section first establishes mean square consensus and then studies the almost sure consensus.

	\section{Consensus}\label{sec:convergence}
	In this section, we investigate the consensus behavior of the delayed
	privacy-preserving algorithm. We begin with mean square weak consensus, and
	then establish mean square strong consensus and almost sure strong consensus
	under stronger summability conditions.

\begin{theorem}\label{thm:dt_ms_weak}
		Assume that the communication graph contains a spanning tree. Suppose that
		\eqref{eq:cbar_A} holds for each nonzero eigenvalue $\lambda_i$,
		$i=2,\dots,N$, of $L$. For each $\lambda_i$, let
		$\varrho_{\lambda_i}>0$ be chosen according to
		Lemma~\ref{lem:dt_fundamental_exp_A} with $\lambda=\lambda_i$, and set
		$
		\varrho_0:=\min_{2\le i\le N}\varrho_{\lambda_i}>0.
		$
		If the following condition holds:
		\[
		\textup{(C1)}\hspace{-0.05cm}
		\sum_{k=0}^{\infty}c(k)\hspace{-0.05cm}=\hspace{-0.05cm}\infty,\hspace{0.05cm}
		\lim_{k\to\infty}\hspace{-0.05cm}
		\sum_{s=0}^{k-1}
		e^{-\varrho_0\sum_{r=s+1}^{k-1}c(r)}
		\hspace{-0.05cm}c(s)^2\bar b(s)^2\hspace{-0.05cm}=\hspace{-0.05cm}0,
		\]
		then the system \eqref{eq:system} under  the protocol \eqref{eq:protocol}
		achieves mean square weak consensus.
	\end{theorem}
	
	\begin{proof}
	To write the closed-loop dynamics compactly, define
	\[
	\begin{aligned}
	x(k)&:=[x_1(k)^\top,\dots,x_N(k)^\top]^\top\in\mathbb R^{nN},\\
	w(k)&:=[w_1(k)^\top,\dots,w_N(k)^\top]^\top\in\mathbb R^{nN}.
	\end{aligned}
	\]
	Then, the overall system can be written as
	\begin{equation}\label{eq:stack_system}
		x(k+1)
		=
		x(k)-c(k)(L\otimes I_n)x(k-d)
		+c(k)(A\otimes I_n)w(k).
	\end{equation}
	
	To analyze consensus, let
	$
	J_N:=\mathbf 1_N\pi^\top
	$
	and
	$
	\delta(k):=\bigl((I_N-J_N)\otimes I_n\bigr)x(k).
	$
	Since $(I_N-J_N)L=L$ and $L(I_N-J_N)=L$, multiplying
	\eqref{eq:stack_system} by $(I_N-J_N)\otimes I_n$ gives
	\begin{equation}\label{eq:delta_dynamics}
		\delta(k+1)
		=
		\delta(k)-c(k)(L\otimes I_n)\delta(k-d)
		+c(k)\tilde\eta(k),
	\end{equation}
	where
	$
	\tilde\eta(k):=\bigl((I_N-J_N)\otimes I_n\bigr)(A\otimes I_n)w(k).
	$
	Thus, consensus can be studied through the asymptotic behavior of the
	disagreement process $\delta(k)$. Define $
	\delta_e(k):=(Q^{-1}\otimes I_n)\delta(k)=\begin{bmatrix}
		0 & \bar\delta^\top(k)
	\end{bmatrix}^\top,
	$
	where $\bar\delta(k)\in\mathbb R^{n(N-1)}$. Then, by Lemma~\ref{lem:lap_decomp}, \eqref{eq:delta_dynamics} becomes
	\begin{equation}\label{eq:dt_bar_delta}
		\bar\delta(k+1)=\bar\delta(k)-c(k)(L_e\otimes I_n)\bar\delta(k-d)+c(k)\bar\eta(k),
	\end{equation}
	where $\bar\eta(k)$ is the lower block of $(Q^{-1}\otimes I_n)\tilde\eta(k)$.
	
	Let $R$ be a nonsingular matrix such that
	$
	RL_eR^{-1}=J,
	$
	where $J$ is the Jordan canonical form of $L_e$ with
	$
	J=\mathrm{diag}(J_{\lambda_2,n_2},\dots,J_{\lambda_\ell,n_\ell}),
	\sum_{p=2}^{\ell}n_p=N-1.
	$
	Define
	$
	Y(k):=(R\otimes I_n)\bar\delta(k).
	$
	Then, \eqref{eq:dt_bar_delta} is transformed into
	\begin{equation}\label{eq:dt_Y}
		Y(k+1)=Y(k)-c(k)(J\otimes I_n)Y(k-d)+c(k)\hat\eta(k),
	\end{equation}
	where $\hat\eta(k):=(R\otimes I_n)\bar\eta(k)$.
	
	Consider one Jordan block associated with an eigenvalue $\lambda\in\mathbb C$ of size $m$, and fix one of the $n$ state coordinates and let $y_r(k)\in\mathbb C$
denote its $r$th Jordan component. The same argument applies to each
state coordinate. Then, \eqref{eq:dt_Y} yields
	\begin{align}
		y_m(k+1) &=\hspace{-0.05cm} y_m(k) \hspace{-0.05cm}- \hspace{-0.05cm}c(k)\lambda y_m(k-d)\hspace{-0.05cm} + \hspace{-0.05cm}c(k)\zeta_m(k), \label{eq:dt_block_last} \\
		y_r(k+1)&=y_r(k)-c(k)\lambda y_r(k-d) \notag\\
		&\quad-c(k)y_{r+1}(k-d)+c(k)\zeta_r(k),
		\quad r=1,\dots,m-1. \label{eq:dt_block_chain}
	\end{align}
	Each $\zeta_r(k)$ is a complex-valued linear combination of the corresponding
coordinates of $\{w_j(k)\}_{j=1}^N$, with deterministic coefficients.
Thus its second moment is bounded by a constant times $\bar b(k)^2$. 
	
	It therefore suffices to analyze the scalar delayed stochastic recursion
	\[
	y(k+1)=y(k)-c(k)\lambda\,y(k-d)+c(k)\zeta(k),
	\]
	where $\Re(\lambda)>0$ and $\mathbb E[\zeta(k)]=0$. 
	The solution admits the variation-of-constants representation
	\begin{equation}\label{eq:dt_voc}
		y(k)=\sum_{r=-d}^{0}\Gamma_h(k,r)y(r)
		+\sum_{s=0}^{k-1}\Gamma(k,s+1)c(s)\zeta(s).
	\end{equation}
	Then, we have
	\begin{align}
		\mathbb{E}\|y(k)\|^2 \le 2\mathbb{E}\bigg\|\sum_{r=-d}^{0}\Gamma_h(k,r)y(r)\bigg\|^2  + 2\mathbb{E}\bigg\|\sum_{s=0}^{k-1}\Gamma(k,s+1)c(s)\zeta(s)\bigg\|^2. \label{eq:dt_ms_bound0}
	\end{align}
	Since the initial history is finite, Cauchy--Schwarz gives
\[
\mathbb E\left\|\sum_{r=-d}^{0}\Gamma_h(k,r)y(r)\right\|^2
\le
(d+1)\sum_{r=-d}^{0}|\Gamma_h(k,r)|^2\,\mathbb E\|y(r)\|^2.
\]
By Lemma~\ref{lem:dt_fundamental_exp_A}, for suitable
$C_h,\varrho_h>0$ this implies
\[
\mathbb E\left\|\sum_{r=-d}^{0}\Gamma_h(k,r)y(r)\right\|^2
\le
(d+1)C_h e^{-\varrho_h\sum_{i=0}^{k-1}c(i)}
\sum_{r=-d}^{0}\mathbb E\|y(r)\|^2.
\]

Since $\sum_{k=0}^{\infty} c(k) = \infty$, it follows that
	\begin{equation}\label{eq:dt_history_zero}
		\lim_{k\to\infty}
		\mathbb E\Big\|\sum_{r=-d}^{0}\Gamma_h(k,r)y(r)\Big\|^2
		=0.
	\end{equation}
	For the stochastic term, temporal uncorrelatedness and zero means give
\begin{align*}
&\mathbb E\left\|
\sum_{s=0}^{k-1}\Gamma(k,s+1)c(s)\zeta(s)
\right\|^2\\
&\quad=
\sum_{s=0}^{k-1}|\Gamma(k,s+1)|^2c(s)^2
\mathbb E\|\zeta(s)\|^2\\
&\quad\le
M_b\sum_{s=0}^{k-1}|\Gamma(k,s+1)|^2c(s)^2\bar b(s)^2,
\end{align*}

where $M_b>0$ is a constant. Applying Lemma~\ref{lem:dt_fundamental_exp_A} again yields the rigorous upper bound:
	\begin{equation}\label{eq:dt_C4_bound_joint}
		\begin{split}
			\mathbb E\Big\|\sum_{s=0}^{k-1}\Gamma(k,s+1)c(s)\zeta(s)\Big\|^2 
		\le
			C M_b\sum_{s=0}^{k-1}\exp\!\Big(-\varrho_0\sum_{r=s+1}^{k-1}c(r)\Big)c(s)^2 \bar b(s)^2.
		\end{split}
	\end{equation}
	By \textup{(C1)} , the right-hand side tends to zero as $k\to\infty$. Combining \eqref{eq:dt_ms_bound0}, \eqref{eq:dt_history_zero},  and \eqref{eq:dt_C4_bound_joint}, we conclude that
	$
	\lim_{k\to\infty}\mathbb E\|y(k)\|^2=0.
	$
	
	It remains to treat the preceding components in the same Jordan chain. For $r=m-1,\dots,1$, the variation-of-constants formula yields:
	\begin{equation}\label{eq:yr_voc}
		\begin{split}
			y_r(k) =\sum_{s=-d}^{0}\Gamma_h(k,s)y_r(s)- \sum_{s=0}^{k-1} \Gamma(k,s+1)c(s)y_{r+1}(s-d) 
			+ \sum_{s=0}^{k-1} \Gamma(k,s+1)c(s)\zeta_r(s).
		\end{split}
	\end{equation}
	Since $\Gamma_h(k,s)$ decays with the cumulative-gain estimate in \eqref{eq:Gammah_exp_A}, the initial condition response trivially satisfies $\lim_{k\to\infty}\mathbb E\Big\|\sum_{s=-d}^{0}$ $\Gamma_h(k,s)y_r(s)\Big\|^2 = 0.$
	For the noise term, the uncorrelatedness of $\zeta_r(s)$ and the bound $\mathbb E\|\zeta_r(s)\|^2 \le M_b \bar b(s)^2$ imply
	\(\mathbb E\|\sum_{s=0}^{k-1}\Gamma(k,s+1)c(s)\zeta_r(s)\|^2
	\le
	C M_b\sum_{s=0}^{k-1}
	e^{-\varrho_0\sum_{l=s+1}^{k-1}c(l)}
	c(s)^2 \bar b(s)^2,\)
	which converges to $0$ as $k\to\infty$ by \textup{(C1)}.
	For the coupling term, define $a_{k,s}:=|\Gamma(k,s+1)|c(s)$. By  Lemma~\ref{lem:dt_kernel_toeplitz}, there exists $M_a<\infty$ such that $\sup_k \sum_{s=0}^{k-1}a_{k,s}\le M_a$. Applying the Cauchy-Schwarz inequality gives:
	\begin{align*}
	&\mathbb E\Big\|\sum_{s=0}^{k-1}
	\Gamma(k,s+1)c(s)y_{r+1}(s-d)\Big\|^2\\
	&\quad\le \mathbb E\bigg[
	\Big(\sum_{s=0}^{k-1}a_{k,s}\|y_{r+1}(s-d)\|\Big)^2\bigg]\\
	&\quad\le M_a\sum_{s=0}^{k-1}a_{k,s}
	\mathbb E\|y_{r+1}(s-d)\|^2.
	\end{align*}
	Since $\lim_{s\to\infty}\mathbb E\|y_{r+1}(s-d)\|^2 = 0$, $\lim_{k\to\infty}a_{k,s}=0$, and $\sum_{s}a_{k,s}\le M_a$, the Toeplitz Lemma (\cite{Knopp:51})  ensures that the right-hand side converges to $0$.
	
	Consequently, $\lim_{k\to\infty}\mathbb E\|y_r(k)\|^2 = 0$. By induction, all Jordan coordinates converge to zero in mean square, implying $\delta(k)\to 0$. The proof is complete.	\end{proof}
	
	\begin{remark}
		Condition \textup{(C1)} may appear complicated because of its convolution structure. In fact, it corresponds to the discrete-time analogue of the sufficient conditions for mean square stability of continuous-time consensus dynamics with additive noises in \cite{zong2018consensus,zong2019consensus}.  In some important special cases, it can be further reduced to simpler sufficient conditions.
	\end{remark}
	
	\begin{remark}\label{d0}
		For delay-free networks ($d=0$), the state-of-the-art work \cite{erfenwjm} establishes mean square weak consensus under the conditions $\sum_{k=0}^{\infty} c(k)=\infty$ and $\lim_{k\to\infty} c(k)\bar b(k)^2 = 0$.
		By contrast, the present paper relaxes this requirement by replacing it with the weaker convolution-type condition \textup{(C1)}. For example, let $c(k)=1/(k+1)$ and $\bar b(k)^2=1+(k+1)\mathbf 1_{\{k=2^m,\ m=0,1,\dots\}}$. Then, we have $\sum_{k=0}^\infty c(k)=\infty$, but $\lim_{k\to\infty}c(k)\bar b(k)^2\neq 0$, so the condition in \cite{erfenwjm} fails. On the other hand, one can verify that $\sum_{s=0}^{k-1}e^{-\varrho_0\sum_{r=s+1}^{k-1}c(r)}c(s)^2\bar b(s)^2 \to 0$, and hence condition \textup{(C1)} holds.
		
		Moreover, the present analysis is established for general directed graphs containing a spanning tree, whereas \cite{erfenwjm} is restricted to undirected graphs. In addition, when the graph is undirected and $d=0$, condition \eqref{eq:cbar_A} becomes $\sup_{k\ge k_0} c(k)<2/\lambda_N(L)$, which is weaker than the condition $\sup_{k\ge0} c(k)\le 1/\lambda_N(L)$ adopted in the aforementioned work.
	\end{remark}

	Theorem~\ref{thm:dt_ms_weak} directly yields the following  corollary. 
	
	\begin{corollary}\label{cor:dt_ms_weak_D1D3}
		Assume that the communication graph contains a spanning tree. Suppose that
		\eqref{eq:cbar_A} holds for each nonzero eigenvalue $\lambda_i$,
		$i=2,\dots,N$, of $L$ with $\lambda=\lambda_i$. If one of the following conditions holds:
		\begin{align*}
			\textup{(C2)}\hspace{0.1cm}
			&\sup_{k\ge0}\bar b(k)\le M_b<\infty,\hspace{0.12cm}
			\sum_{k=0}^{\infty}c(k)=\infty,\hspace{0.12cm}
			\lim_{k\to\infty}c(k)=0;\\[1mm]
			\textup{(C3)}\hspace{0.1cm}
			&c(k)\equiv c_0>0,\quad
			\lim_{k\to\infty}\bar b(k)=0,
		\end{align*}
		where $M_b>0$ and $c_0>0$ are constants, then the system
		\eqref{eq:system} under the protocol \eqref{eq:protocol} achieves mean square
		weak consensus.
	\end{corollary}
	
	\begin{proof}
The small-gain condition is assumed in both cases. For condition
\textup{(C2)}, it is also automatic eventually, since $c(k)\to0$.
Choose the common positive squared decay rate $\varrho_0$ as in
Theorem~\ref{thm:dt_ms_weak}.

Under \textup{(C2)}, set
\[
\widetilde a_{k,s}
=
\exp\!\left(-\varrho_0\sum_{r=s+1}^{k-1}c(r)\right)c(s).
\]
The telescoping estimate in Appendix B, with $\mu=\varrho_0$, gives
uniformly bounded row sums, and each fixed column tends to zero.
Since $c(s)\bar b(s)^2\to0$, the Toeplitz argument yields
\[
\sum_{s=0}^{k-1}
\exp\!\left(-\varrho_0\sum_{r=s+1}^{k-1}c(r)\right)
c(s)^2\bar b(s)^2\longrightarrow0.
\]

Under \textup{(C3)}, put $q=\exp(-\varrho_0c_0)\in(0,1)$.
The same convolution becomes
\[
c_0^2\sum_{s=0}^{k-1}q^{\,k-s-1}\bar b(s)^2.
\]
Each fixed-column coefficient tends to zero, while the row sums are
bounded by $c_0^2/(1-q)$. Since $\bar b(s)^2\to0$, this expression
also tends to zero. Equivalently, one may split at a fixed large
index and bound the tail by a geometric series.
Thus \textup{(C1)} holds in either case, and
Theorem~\ref{thm:dt_ms_weak} proves the assertion.
\end{proof}

	\begin{remark}
		Condition \textup{(C1)} extends the diminishing gain condition
		used in a representative stochastic-consensus result.
		Specifically, \cite{lizhang2010} established mean square weak consensus
		for MASs with stochastic communication noises under the
		condition recovered by  
		\textup{(C2)}. In contrast, the present paper allows communication delays
		and privacy noises, and adopts the more general convolution-type condition
		\textup{(C1)} as the basic requirement
		for mean square weak consensus.
	\end{remark}

	We next strengthen the above result from weak consensus to strong consensus
	in both the mean square and almost sure senses.   
	
	\begin{theorem}
		\label{thm:dt_strong_consensus}
		Assume that the communication graph contains a spanning tree. Suppose that
		\eqref{eq:cbar_A} holds for each nonzero eigenvalue $\lambda_i$,
		$i=2,\dots,N$, of $L$ with $\lambda=\lambda_i$. If the following condition
		holds:
		\[
		\textup{(C4)}\quad
		\sum_{k=0}^{\infty} c(k)=\infty,\quad
		\sum_{k=0}^{\infty} c(k)^2\bar b(k)^2<\infty,
		\]
		then the system \eqref{eq:system} under the protocol \eqref{eq:protocol}
		achieves mean square strong consensus and almost sure strong consensus.
	\end{theorem}
	
		  \begin{proof}
	First, we prove mean square strong consensus. Let $\pi\in\mathbb R^N$ satisfy
	$	\pi^T L=0, \pi^T\mathbf 1_N=1,
	$
	and define $
	J_N:=\mathbf 1_N\pi^T,
	\bar x(k):=(\pi^T\otimes I_n)x(k),
	\delta(k):=((I_N-J_N)\otimes I_n)x(k).
	$
	Then, we have $x(k)=(\mathbf 1_N\otimes I_n)\bar x(k)+\delta(k)$.
	We first prove the mean-square convergence of the consensus component.
	Multiplying \eqref{eq:stack_system} by $(\pi^T\otimes I_n)$ and using
	$\pi^T L=0$, we obtain
	\begin{equation}\label{eq:xbar_rec_fixed}
		\bar x(k+1)=\bar x(k)+c(k)\bar\eta(k),
	\end{equation}
	where $\bar\eta(k):=(\pi^T A\otimes I_n)w(k)$.
	Thus, iterating \eqref{eq:xbar_rec_fixed} gives
	\begin{equation}\label{eq:xbar_sum_fixed}
		\bar x(k)=\bar x(0)+\sum_{s=0}^{k-1}c(s)\bar\eta(s).
	\end{equation}
	Let
	\[
	\mathcal F_k:=
	\sigma\bigl(x(r),-d\le r\le 0;\, w(t),0\le t\le k-1\bigr),
	\quad k\ge0.
	\]
	By the zero-mean property of the Laplace noises and their independence across
	time, $\{\bar\eta(k)\}$ is a martingale-difference sequence with respect to
	$\{\mathcal F_k\}$, namely $\mathbb E[\bar\eta(k)\mid\mathcal F_k]=0$.
	Moreover, since $\bar\eta(k)$ is a deterministic linear combination of the
	entries of $w(k)$, and $\mathbb E|w_j^{(\ell)}(k)|^2=2b_j(k)^2\le 2\bar b(k)^2$,
	there exists a constant $C_\pi>0$, depending only on $\pi$, $A$, and $n$,
	such that
	\begin{equation}\label{eq:bar_eta_bd_fixed}
		\mathbb E\bigl[\|\bar\eta(k)\|^2\mid\mathcal F_k\bigr]
		\le C_\pi \bar b(k)^2,
		\qquad \text{a.s.}
	\end{equation}
	Define $M(k):=\sum_{s=0}^{k-1}c(s)\bar\eta(s)$, and then $\{M(k),\mathcal F_k\}$ is an $\mathbb R^n$-valued martingale. Furthermore,
	by the orthogonality of martingale increments and \eqref{eq:bar_eta_bd_fixed},
	\begin{align*}
		\mathbb E\|M(k)\|^2
		=
		\sum_{s=0}^{k-1}c(s)^2\mathbb E\|\bar\eta(s)\|^2
		\le
		C_\pi\sum_{s=0}^{k-1}c(s)^2\bar b(s)^2
		\le
		C_\pi\sum_{s=0}^{\infty}c(s)^2\bar b(s)^2
		<\infty.
	\end{align*}
	Hence $\{M(k)\}$ is bounded in $L^2$. By the $L^2$ martingale convergence
	theorem, there exists an $\mathbb R^n$-valued random vector $M(\infty)\in L^2$
	such that $\lim_{k\to\infty}\mathbb E\|M(k)-M(\infty)\|^2=0$.
	Consequently, with $x^\star:=\bar x(0)+M(\infty)$, we obtain from
	\eqref{eq:xbar_sum_fixed} that
	\begin{equation}\label{eq:xbar_to_xstar_fixed}
		\lim_{k\to\infty}\mathbb E\|\bar x(k)-x^\star\|^2=0.
	\end{equation}
	We next show that the disagreement component converges to zero in mean square.
	By Theorem~\ref{thm:dt_ms_weak}, it suffices to verify  \textup{(C1)}.
	Let $q(s):=c(s)^2\bar b(s)^2$. Then, we have $\sum_{s=0}^{\infty}q(s)<\infty$.
	For any $\varepsilon>0$, choose $N_\varepsilon$ such that
	$\sum_{s=N_\varepsilon}^{\infty}q(s)<\varepsilon$.
	For the finite part, since $\sum_{k=0}^{\infty}c(k)=\infty$, for every fixed
	$s$ one has
	$\exp\!\Bigl(-\varrho_0\sum_{r=s+1}^{k-1}c(r)\Bigr)\to0$ as $k\to\infty$.
	Therefore,
	\[
	\sum_{s=0}^{N_\varepsilon-1}
	\exp\!\Bigl(-\varrho_0\sum_{r=s+1}^{k-1}c(r)\Bigr)q(s)
	\to0.
	\]
	For the tail part, using the fact that the exponential factor is bounded by
	one, we have
	\[
	\sum_{s=N_\varepsilon}^{k-1}
	\exp\!\Bigl(-\varrho_0\sum_{r=s+1}^{k-1}c(r)\Bigr)q(s)
	\le
	\sum_{s=N_\varepsilon}^{\infty}q(s)
	<\varepsilon.
	\]
	Hence
	\[
	\limsup_{k\to\infty}
	\sum_{s=0}^{k-1}
	\exp\!\Bigl(-\varrho_0\sum_{r=s+1}^{k-1}c(r)\Bigr)c(s)^2\bar b(s)^2
	\le \varepsilon.
	\]
	Letting $\varepsilon\downarrow0$ yields \textup{(C1)}. Therefore,
	Theorem~\ref{thm:dt_ms_weak} gives
	\begin{equation}\label{eq:delta_ms_to0_fixed}
		\lim_{k\to\infty}\mathbb E\|\delta(k)\|^2=0.
	\end{equation}
	Finally, for each Agent $i$,
	$x_i(k)-x^\star=\delta_i(k)+\bar x(k)-x^\star$,
	where $\delta_i(k)$ denotes the $i$th $n$-dimensional block of $\delta(k)$.
	Thus,
	\[
	\begin{aligned}
		\mathbb E\|x_i(k)-x^\star\|^2
		&\le
		2\mathbb E\|\delta_i(k)\|^2
		+
		2\mathbb E\|\bar x(k)-x^\star\|^2 \\
		&\le
		2\mathbb E\|\delta(k)\|^2
		+
		2\mathbb E\|\bar x(k)-x^\star\|^2 .
	\end{aligned}
	\]
	Taking $k\to\infty$ and using \eqref{eq:xbar_to_xstar_fixed} and
	\eqref{eq:delta_ms_to0_fixed}, we obtain
	$\lim_{k\to\infty}\mathbb E\|x_i(k)-x^\star\|^2=0$,
	$i=1,\dots,N$.
	This proves mean-square strong consensus.
	
	Then, we prove almost sure strong consensus.
We use the same filtration $\{\mathcal F_k\}$ defined above, which
includes the entire initial history and only the noises at times
$0,\dots,k-1$.
	By Lemma~\ref{lem:lap_decomp}, there exists a probability vector
	$\pi\in\mathbb R^N$ such that $\pi^\top L=0$ and $\pi^\top \mathbf 1_N=1$.
	Define
	$
	J:=\mathbf 1_N\pi^\top,
	\bar x(k):=(\pi^\top\otimes I_n)x(k),
	\delta(k):=((I_N-J)\otimes I_n)x(k).
	$
	Then, we have
	$
	x(k)=(\mathbf 1_N\otimes I_n)\bar x(k)+\delta(k).
	$
	We first prove that
	\begin{equation}\label{eq:delta_as_zero}
		\delta(k)\to0,\qquad \text{a.s.}
	\end{equation}
	Using the same Jordan decomposition as in the proof of
	Theorem~\ref{thm:dt_ms_weak}, it suffices to consider one Jordan block
	associated with a nonzero eigenvalue $\lambda\in\mathbb C$ with
	$\Re(\lambda)>0$. Let the corresponding delayed coordinates be
	$y_1(k),\dots,y_m(k)$ for one fixed state coordinate; these variables
are complex scalars. Then,
	\begin{align}
		y_m(k+1)
		&=y_m(k)-c(k)\lambda\,y_m(k-d)+c(k)\omega_m(k),
		\label{eq:as_block_last_strong}\\
		y_r(k+1)
		&=y_r(k)-c(k)\lambda\,y_r(k-d) \notag\\
		&\quad-c(k)y_{r+1}(k-d)+c(k)\omega_r(k),
		\qquad r=1,\dots,m-1.
		\label{eq:as_block_chain_strong}
	\end{align}
	where each $\omega_r(k)$ is a deterministic linear combination of
	$\{w_j(k)\}_{j=1}^N$.
	By the Laplace noise model and the independence across time, there exists
	a constant $C_r>0$ such that
	\begin{equation}\label{eq:omega_second_strong}
		\mathbb E\!\left[\omega_r(k)\mid \mathcal F_k\right]=0,
		\qquad
		\mathbb E\!\left[|\omega_r(k)|^2\mid \mathcal F_k\right]\le C_r\bar b(k)^2
		\quad \text{a.s.}
	\end{equation}
	
	Now introduce the corresponding delay-free comparison block.
For the later delayed comparison, define its finite history by
$\zeta_r(j)=y_r(j)$ for $-d\le j\le0$, $r=1,\dots,m$, and evolve
from $k=0$ according to
	\begin{align}
		\zeta_m(k+1)
		&=\zeta_m(k)-c(k)\lambda\,\zeta_m(k)+c(k)\omega_m(k),
		\label{eq:as_zeta_last_strong}\\
		\zeta_r(k+1)
		&=\zeta_r(k)-c(k)\lambda\,\zeta_r(k)-c(k)\zeta_{r+1}(k) +c(k)\omega_r(k),
		\qquad r=1,\dots,m-1.
		\label{eq:as_zeta_chain_strong}
	\end{align}
	
	We claim that
	\begin{equation}\label{eq:zeta_as_zero}
		\zeta_r(k)\to0,\qquad \text{a.s.},\qquad r=1,\dots,m.
	\end{equation}
	
	For the last coordinate, let
	$\bar c:=\sup_{k\ge k_0}c(k)$ and
	$\beta_\lambda:=2\Re(\lambda)-|\lambda|^2\bar c$.
	Since \eqref{eq:cbar_A} gives
	$2\Re(\lambda)-(2d+1)|\lambda|^2\bar c>0$,
	one has $\beta_\lambda>0$. Hence, for all $k\ge k_0$,
	\[
	|1-c(k)\lambda|^2
	=
	1-\bigl(2\Re(\lambda)-|\lambda|^2c(k)\bigr)c(k)
	\le
	1-\beta_\lambda c(k).
	\]
	Therefore, for all $k\ge k_0$,
	\begin{align*}
		\mathbb E\!\left(|\zeta_m(k+1)|^2\mid\mathcal F_k\right)
		&=
		|1-c(k)\lambda|^2|\zeta_m(k)|^2 +
		c(k)^2\mathbb E\!\left(|\omega_m(k)|^2\mid\mathcal F_k\right)\\
		&\le
		(1-\beta_\lambda c(k))|\zeta_m(k)|^2
		+
		C\,c(k)^2\bar b(k)^2 .
	\end{align*}
	where the cross term vanishes by \eqref{eq:omega_second_strong}.
	Since $\sum_{k=0}^{\infty}c(k)^2\bar b(k)^2<\infty$,
	the Robbins--Siegmund theorem (\cite{robbinssiegmund1971}) implies that
	$|\zeta_m(k)|^2$ converges almost surely and
	\[
	\sum_{k=k_0}^{\infty}c(k)|\zeta_m(k)|^2<\infty
	\quad\text{a.s.}
	\]
	Because $\sum_{k=0}^{\infty}c(k)=\infty$, the only possible almost sure limit is zero. Thus
	$\zeta_m(k)\to0$ a.s.
	We next prove \eqref{eq:zeta_as_zero} by backward induction. Fix
	$r\in\{1,\dots,m-1\}$ and assume that $\zeta_{r+1}(k)\to0$ a.s.
	Let $\widetilde\zeta_r(k)$ be the solution of
	\[
	\widetilde\zeta_r(k+1)=\widetilde\zeta_r(k)-c(k)\lambda\,\widetilde\zeta_r(k)+c(k)\omega_r(k),
\quad \widetilde\zeta_r(0)=\zeta_r(0).
	\]
	By the same Robbins--Siegmund argument as above,
	$\widetilde\zeta_r(k)\to0$ a.s.
	Now define $\vartheta_r(k):=\zeta_r(k)-\widetilde\zeta_r(k)$, we have
	\[
	\vartheta_r(k+1)=\vartheta_r(k)-c(k)\lambda\,\vartheta_r(k)-c(k)\zeta_{r+1}(k).
	\]
	By the variation-of-constants formula,
	\[
	\vartheta_r(k)
	=
	\Phi_\lambda(k,0)\vartheta_r(0)
	-
	\sum_{s=0}^{k-1}\Phi_\lambda(k,s+1)c(s)\zeta_{r+1}(s),
	\]
	where
	$\Phi_\lambda(k,s):=\prod_{j=s}^{k-1}(1-c(j)\lambda)$, $k>s$, and
	$\Phi_\lambda(k,k)=1$.
	For $k>s\ge k_0$,
	\[
	|\Phi_\lambda(k,s+1)|^2
	\le
	\exp\!\Bigl(-\beta_\lambda\sum_{j=s+1}^{k-1}c(j)\Bigr).
	\]
	Hence the proof of Lemma~\ref{lem:dt_kernel_toeplitz} applies verbatim
	with $\Gamma(k,s+1)$ replaced by $\Phi_\lambda(k,s+1)$.
	Since $\zeta_{r+1}(s)\to0$ a.s., the Toeplitz Lemma (\cite{Knopp:51}) implies
	\[
	\sum_{s=0}^{k-1}\Phi_\lambda(k,s+1)c(s)\zeta_{r+1}(s)\to0
	\qquad \text{a.s.}
	\]
	Also, $\Phi_\lambda(k,0)\vartheta_r(0)\to0$.
	Therefore, $\vartheta_r(k)\to0$ a.s.
	and hence $\zeta_r(k)=\widetilde\zeta_r(k)+\vartheta_r(k)\to0$ a.s.
	By backward induction, \eqref{eq:zeta_as_zero} follows.
	
	Next, define $\theta_r(k):=\zeta_r(k)-y_r(k)$, $r=1,\dots,m$.
	For the last coordinate, subtracting \eqref{eq:as_block_last_strong}
	from \eqref{eq:as_zeta_last_strong} yields
	\[
	\theta_m(k+1)=\theta_m(k)-c(k)\lambda\,\theta_m(k-d)+c(k)g_m(k),
	\]
	where
	$g_m(k):=\lambda\bigl(\zeta_m(k-d)-\zeta_m(k)\bigr)\to0$ a.s.
	By the variation-of-constants formula for delayed difference equations,
	\[
	\theta_m(k)
	=
	\sum_{r=-d}^{0}\Gamma_h(k,r)\theta_m(r)
	+
	\sum_{s=0}^{k-1}\Gamma(k,s+1)c(s)g_m(s).
	\]
	The first term converges to zero by \eqref{eq:Gammah_exp_A} and
	$\sum_{k=0}^{\infty}c(k)=\infty$, while the second term converges to zero by
	Lemma~\ref{lem:dt_kernel_toeplitz}, since $g_m(k)\to0$ a.s. Hence
	$\theta_m(k)\to0$ a.s.
	Together with $\zeta_m(k)\to0$ a.s., this yields
	$y_m(k)\to0$ a.s.
	Now suppose that $y_{r+1}(k)\to0$ a.s. for some $r\in\{1,\dots,m-1\}$.
	Subtracting \eqref{eq:as_block_chain_strong} from
	\eqref{eq:as_zeta_chain_strong}, we obtain
	\[
	\theta_r(k+1)=\theta_r(k)-c(k)\lambda\,\theta_r(k-d)+c(k)h_r(k),
	\]
	where
	$
	h_r(k):=
	\lambda\bigl(\zeta_r(k-d)-\zeta_r(k)\bigr)
	-
	\bigl(\zeta_{r+1}(k)-y_{r+1}(k-d)\bigr)\to0
	\hspace{0.2cm} \text{a.s.}
	$
	Again, by the variation-of-constants formula and
	Lemma~\ref{lem:dt_kernel_toeplitz},
	$\theta_r(k)\to0$ a.s.
	and therefore
	$y_r(k)\to0$ a.s.
	By backward induction, all Jordan coordinates satisfy
	$y_r(k)\to0$ a.s., $r=1,\dots,m$.
	Since this holds for every Jordan block, the disagreement process satisfies
	$\delta(k)\to0$ a.s.
	This proves \eqref{eq:delta_as_zero}.
	
	It remains to prove the almost sure convergence of the consensus component.
	Multiplying \eqref{eq:stack_system} by $(\pi^\top\otimes I_n)$ and using
	$\pi^\top L=0$, we obtain
	\[
	\bar x(k+1)=\bar x(k)+c(k)(\pi^\top A\otimes I_n)w(k).
	\]
	Thus $\{\bar x(k)\}$ is a martingale with respect to $\{\mathcal F_k\}$.
	Moreover, since the entries of $w(k)$ are independent Laplace random variables,
	there exists a constant $C_\pi>0$ such that
	\[
	\mathbb E\!\left[\|(\pi^\top A\otimes I_n)w(k)\|^2\mid\mathcal F_k\right]
	\le
	C_\pi \bar b(k)^2
	\qquad \text{a.s.}
	\]
	Hence
	\[
	\mathbb E\!\left[\|\bar x(k+1)\|^2\mid\mathcal F_k\right]
	\le
	\|\bar x(k)\|^2 + C_\pi c(k)^2\bar b(k)^2.
	\]
	Since $\sum_{k=0}^{\infty}c(k)^2\bar b(k)^2<\infty$,
	it follows that $\sup_{k\ge0}$ $\mathbb E\|\bar x(k)\|^2<\infty$.
	Therefore, by the martingale convergence theorem, there exists an
	$\mathbb R^n$-valued random vector $x^\star$ such that
	$\bar x(k)\to x^\star$ a.s.
	Finally, using $x(k)=(\mathbf 1_N\otimes I_n)\bar x(k)+\delta(k)$,
	together with \eqref{eq:delta_as_zero}, we conclude that
	$x(k)\to (\mathbf 1_N\otimes I_n)x^\star$ a.s.
	Equivalently, $x_i(k)\to x^\star$ a.s., $i=1,\dots,N$.
	Therefore, the system achieves almost sure strong consensus.
	Now, the proof is completed.
	\end{proof}

	\begin{remark}
Condition \textup{(C4)} implies \textup{(C1)} by the finite-prefix
and tail estimate in the preceding proof. The converse need not hold.
For example, take $c(k)\equiv1$ and $\bar b(k)=(k+1)^{-1/2}$.
Then
\[
\sum_{k=0}^\infty c(k)^2\bar b(k)^2
=\sum_{k=0}^\infty(k+1)^{-1}=\infty,
\]
whereas the convolution in \textup{(C1)} tends to zero by the
bounded geometric-kernel argument used for condition \textup{(C3)}.
\end{remark}
	
	As direct consequences of Theorem~\ref{thm:dt_strong_consensus}, we obtain
	the following two corollaries.
	
	\begin{corollary}
		\label{cor:dt_strong_special}
		Assume that the communication graph contains a spanning tree. Suppose that
		\eqref{eq:cbar_A} holds for each nonzero eigenvalue $\lambda_i$,
		$i=2,\dots,N$, of $L$ with $\lambda=\lambda_i$. If either of the following
		two conditions is satisfied:
		\begin{align*}
			\textup{(C5)}\hspace{0.1cm}
			&\sup_{k\ge0}\bar b(k)\le M_b<\infty,\hspace{0.12cm}
			\sum_{k=0}^{\infty}c(k)=\infty,\hspace{0.12cm}
			\sum_{k=0}^{\infty}c(k)^2<\infty;\\[1mm]
			\textup{(C6)}\hspace{0.1cm}
			&c(k)\equiv c_0>0,\quad
			\sum_{k=0}^{\infty}\bar b(k)^2<\infty,
		\end{align*}
		where $M_b>0$ and $c_0>0$ are constants, then the system
		\eqref{eq:system} under the protocol \eqref{eq:protocol} achieves mean square
		strong consensus and almost sure strong consensus.
	\end{corollary}
	
	\begin{proof}
	Under either \textup{(C5)} or \textup{(C6)}, one has
	$\sum_{k=0}^{\infty}c(k)=\infty$ and
	$\sum_{k=0}^{\infty}c(k)^2\bar b(k)^2<\infty$, and hence condition
	\textup{(C4)} holds. Therefore, the conclusion follows directly from
	Theorem~\ref{thm:dt_strong_consensus}.
	\end{proof}
	
	\begin{remark}
		Conditions \textup{(C2)}--\textup{(C3)} and
		\textup{(C5)}--\textup{(C6)} cover two representative settings in the
		literature. Conditions \textup{(C2)} and \textup{(C5)} correspond to the
		standard stochastic approximation setting, where the step size is
		diminishing and the noise variance is constant or uniformly bounded; see,
		e.g., \cite{LiZhang:09,huangmanton2010,xuzhangxie2012,wang2025stochastic}.
		Conditions \textup{(C3)} and \textup{(C6)} correspond to the fixed
		step-size setting with decaying privacy-noise variance, as considered in
		\cite{huang2012differentially,nozari2017differentially,fiore2019resilient,
			gao2019differentially,tian2025privacy}. Compared with these works, the
		present analysis further incorporates communication delays and
		time-varying privacy noises into a unified framework.
	\end{remark}
	
	\begin{corollary}
		\label{cor:dt_strong_unbounded}
		Assume that the communication graph contains a spanning tree. Suppose
		the control gain and the privacy noise scale are designed as
		$c(k)=c_0/(k+1)^\beta$ and $\bar b(k)=b_0(k+1)^\gamma$, respectively,
		where $c_0,b_0>0$. If $\beta\in(0,1]$, $\gamma<\beta-1/2$, and
		\eqref{eq:cbar_A} is satisfied for each nonzero eigenvalue $\lambda_i$ 
		of $L$  with $\lambda=\lambda_i$, then the system achieves mean square strong consensus and almost
		sure strong consensus. Consequently, it also achieves mean square weak
		consensus.
	\end{corollary}
	
	\begin{proof}
	Since $\beta\in(0,1]$, we have
	$\sum_{k=0}^{\infty}c(k)=c_0\sum_{k=0}^{\infty}(k+1)^{-\beta}=\infty$.
	Moreover,
	$
	c(k)^2\bar b(k)^2
	=
	c_0^2b_0^2(k+1)^{-2\beta+2\gamma}.
	$
	The condition $\gamma<\beta-1/2$ implies $2\beta-2\gamma>1$, and hence
	$
	\sum_{k=0}^{\infty}c(k)^2\bar b(k)^2
	=
	c_0^2b_0^2
	\sum_{k=0}^{\infty}(k+1)^{-2\beta+2\gamma}
	<\infty .
	$
	Therefore, \textup{(C4)} holds. The
	mean square strong consensus and almost sure strong consensus then follow
	from Theorem~\ref{thm:dt_strong_consensus}. Finally, mean square strong
	consensus implies mean square weak consensus.
	\end{proof}

	\begin{remark}
		Corollary~\ref{cor:dt_strong_unbounded} shows that strong consensus can
		still be achieved even when the privacy noise increases over time. This
		relaxes the commonly imposed requirement that the privacy noise should be
		constant or decay with time; see, e.g.,
		\cite{huang2012differentially,nozari2017differentially,fiore2019resilient,
			gao2019differentially,tian2025privacy}.
	\end{remark}
	
	\begin{remark}
		To the best of our knowledge, the present result provides the first
		mean square and almost sure stability analysis for first-order delayed
		systems with injected Laplace privacy noises. In essence, this addresses
		the stability problem of delayed systems driven by additive noises. The
		analysis is based on the difference resolvent function method and a
		backstepping technique.
	\end{remark}
	
	\section{Privacy analysis}\label{sec:privacy}

This section protects the complete delayed initial history through the
global transcript. 

For a fixed history $H$, let the deterministic observation sequence be
$
\vartheta=(\vartheta(0),$ $\vartheta(1),\dots).
$
Define the reconstructed state $x^{H,\vartheta}$ by its initial history
$H$ and
\begin{equation}\label{eq:observation_driven_state}
	x^{H,\vartheta}(t+1)
	=
	x^{H,\vartheta}(t)
	-c(t)(K\otimes I_n)x^{H,\vartheta}(t-d)
	+c(t)(A\otimes I_n)\vartheta(t).
\end{equation}
This follows by substituting
$w(t)=\vartheta(t)-x^{H,\vartheta}(t-d)$ into
\eqref{eq:stack_system}. The reconstruction is causal:
$x^{H,\vartheta}(t)$ depends only on $H$ and
$\vartheta(0),\dots,\vartheta(t-1)$.
For a finite transcript $\Theta^T$, write
$\rho(H,\Theta^T)(k):=x^{H,\Theta}(k-d)$, $0\le k\le T$.
Negative state indices are supplied by $H$.

\begin{definition}\label{def:sensitivity_delay_global}
	For $k\ge0$, the pointwise history sensitivity is
	\begin{equation}\label{eq:def_pointwise_sens_global}
		S(k):=
		\sup_{H^a\sim H^b}\ 
		\sup_{\vartheta^{k-1}\in(\mathbb R^{nN})^k}
		\left\|
		x^{H^a,\vartheta}(k-d)-x^{H^b,\vartheta}(k-d)
		\right\|_1.
	\end{equation}
	For $k=0$ the observation prefix is empty. By causality, this definition
	does not depend on a terminal horizon or a selected observation event.
\end{definition}

\begin{theorem}\label{thm:history_sensitivity_global}
Suppose that, for some integer $k_0\ge0$,
$\bar c:=\sup_{k\ge k_0}c(k)<\infty$ and
\begin{equation}\label{eq:cbar_priv_global}
\alpha_H:=
\min_{1\le i\le N}
\left\{2\kappa_i-(2d+1)\kappa_i^2\bar c\right\}>0.
\end{equation}
For each $i$, apply Lemma~\ref{lem:dt_fundamental_exp_A} with
$\lambda=\kappa_i$ and choose its squared decay parameter $\eta_i>0$
satisfying \eqref{eq:h_varrho_condition}. Let $C_{h,i}\ge1$ be the history constant specified in
Lemma~\ref{lem:dt_fundamental_exp_A} for $\lambda=\kappa_i$
and $\varrho=\eta_i$, and set
\[
\rho_H:=\frac12\min_i\eta_i,\qquad
M:=\max_i\sqrt{C_{h,i}}\ge1.
\]
Then, for every $k\ge0$,
\begin{equation}\label{eq:S_exp_global}
S(k)\le
M\Delta_H
\exp\!\left(-\rho_H\sum_{s=0}^{k-d-1}c(s)\right).
\end{equation}
The constants are independent of $k$, the terminal horizon, the
adjacent histories, and the observation transcript.
\end{theorem}

\begin{proof}
For two adjacent histories and the same observation prefix, let
$e(t)=x^{H^a,\vartheta}(t)-x^{H^b,\vartheta}(t)$.
Subtracting \eqref{eq:observation_driven_state} gives
\[
e_i(t+1)=e_i(t)-\kappa_ic(t)e_i(t-d),\qquad t\ge0.
\]
If the histories differ only for Agent $i_0$, then $e_i(t)=0$ for
all $t\ge-d$ and $i\ne i_0$. For $u\ge0$, denote the history
fundamental solution associated with $\kappa_i$ by $\Gamma_{h,i}$.
Then
\[
e_{i_0}(u)=
\sum_{r=-d}^0\Gamma_{h,i_0}(u,r)e_{i_0}(r).
\]
Taking square roots of the estimate in
Lemma~\ref{lem:dt_fundamental_exp_A} yields
\[
|\Gamma_{h,i}(u,r)|
\le
\sqrt{C_{h,i}}\,
\exp\!\left(-\frac{\eta_i}{2}\sum_{s=0}^{u-1}c(s)\right).
\]
The triangle inequality and the history adjacency bound therefore give
\eqref{eq:S_exp_global} for $k\ge d+1$, with $u=k-d$.
For $0\le k\le d$, the released state is a single initial-history
entry, so $S(k)\le\Delta_H$, which also satisfies the bound since
$M\ge1$ and the sum is empty.
The state difference does not depend on the observation values.
Taking the suprema in the definition of $S(k)$ completes the proof.
\end{proof}

\begin{remark}\label{rem:exact_history_sensitivity}
Extend $\Gamma_{h,i}(u,r)$ to $-d\le u\le0$ by its defining unit
history. Since $\mathcal H_i=(\mathbb R^n)^{d+1}$, the same argument,
and a difference concentrated in one history coordinate, give the
exact expression
\begin{equation}\label{eq:S_exact_history}
S(k)=\Delta_H
\max_{\substack{1\le i\le N\\-d\le r\le0}}
|\Gamma_{h,i}(k-d,r)|.
\end{equation}
Thus finite-time sensitivities can be computed by scalar recursions.
An infinite-horizon certificate still requires a proven tail bound.
The parameter $\rho_H$ is an amplitude decay rate; it is distinct
from the squared rates used in the convergence analysis.
\end{remark}

We next establish the finite-horizon privacy guarantee.

\begin{theorem}
		\label{thm:history_dp_global}
		For each finite horizon $T\ge0$, the mechanism $\mathcal M^T$ is
$\varepsilon(T)$-differentially private with the valid budget
\eqref{eq:eps_global}, where the sensitivity is defined in
Definition~\ref{def:sensitivity_delay_global}. Namely,
		\begin{equation}\label{eq:eps_global}
			\varepsilon(T)
			=
			\sum_{k=0}^{T}\frac{S(k)}{b_{\min}(k)},
			\quad
			b_{\min}(k):=\min_{1\le j\le N} b_j(k).
		\end{equation}
In particular, whenever \eqref{eq:S_exp_global} holds,
\begin{equation}\label{eq:eps_global_exp}
			\begin{aligned}
				\varepsilon(T)
				&\hspace{-0.05cm}\le\hspace{-0.05cm}
				M\Delta_H
				\sum_{k=0}^{T}
				\frac{1}{b_{\min}(k)}
				\exp\!\left(\hspace{-0.05cm}
				-\rho_H
				\sum_{s=0}^{k-d-1}c(s)\hspace{-0.05cm}
				\right).
			\end{aligned}
		\end{equation}
	\end{theorem}
	
	\begin{proof}
	Fix two adjacent delayed histories $H^a\sim H^b$ and an arbitrary
	observation sequence
	$
	\Theta^T=\{\Theta(k)\}_{k=0}^{T}.
	$
	For fixed $H$, the reconstruction
\eqref{eq:observation_driven_state} makes the delayed state at time
$k$ a function of the strictly earlier observations.
Conditional on those observations, $w(k)$ has the prescribed independent
Laplace density. Consequently, the conditional-density chain rule gives
$p(\Theta^T\mid H)=\prod_{k=0}^T
p(\Theta(k)\mid H,\Theta(0),\dots,\Theta(k-1))$, or explicitly
	\[
	\begin{aligned}
		p(\Theta^T\mid H)
		=
		\prod_{k=0}^{T}
		\prod_{j=1}^{N}
		\frac{1}{\bigl(2b_j(k)\bigr)^n} 
		\exp\!\left(
		-\frac{
			\left\|
			\Theta_j(k)-\rho_j(H,\Theta^T)(k)
			\right\|_1
		}{b_j(k)}
		\right).
	\end{aligned}
	\]
	It follows that
	\[
	\begin{aligned}
		\frac{p(\Theta^T\mid H^a)}
		{p(\Theta^T\mid H^b)}
	\le
		\exp\!\Bigg(
		\sum_{k=0}^{T}
		\sum_{j=1}^{N}
		\frac{1}{b_j(k)}
				\left\|
		\rho_j(H^a,\Theta^T)(k)
		-
		\rho_j(H^b,\Theta^T)(k)
		\right\|_1
		\Bigg).
	\end{aligned}
	\]
	Set $d_j(k):=\rho_j(H^a,\Theta^T)(k)-
	\rho_j(H^b,\Theta^T)(k)$ and
	$d(k):=\operatorname{col}(d_1(k),\dots,d_N(k))$.
	Since $b_j(k)\ge b_{\min}(k)$ for all $j$, we have
	\[
	\begin{aligned}
		\sum_{j=1}^{N}
		\frac{\|d_j(k)\|_1}{b_j(k)}
	&\le \frac{1}{b_{\min}(k)}
		\sum_{j=1}^{N}\|d_j(k)\|_1\\
	&=\frac{\|d(k)\|_1}{b_{\min}(k)}\\
	&\le \frac{S(k)}{b_{\min}(k)} .
	\end{aligned}
	\]
	Thus,
	\[
	\begin{aligned}
		p(\Theta^T\mid H^a)
		\le
		\exp\!\left(
		\sum_{k=0}^{T}
		\frac{S(k)}{b_{\min}(k)}
		\right)
		p(\Theta^T\mid H^b)
		=
		e^{\varepsilon(T)}
		p(\Theta^T\mid H^b).
	\end{aligned}
	\]
	For a measurable set
	$\mathcal O\subseteq(\mathbb R^{nN})^{T+1}$, write
	$P_H^T(\mathcal O):=
	\mathbb P\{\mathcal M^T(H)\in\mathcal O\}$. Then
	\[
	\begin{aligned}
		P_{H^a}^T(\mathcal O)
		&=
		\int_{\mathcal O}
		p(\Theta^T\mid H^a)\,d\Theta^T
		\\
		&\le
		e^{\varepsilon(T)}
		\int_{\mathcal O}
		p(\Theta^T\mid H^b)\,d\Theta^T
		\\
		&=
		e^{\varepsilon(T)}
		P_{H^b}^T(\mathcal O).
	\end{aligned}
	\]
	This proves that $\mathcal M^T$ is $\varepsilon(T)$-differentially
	private. Finally, \eqref{eq:eps_global_exp} follows directly from
	\eqref{eq:S_exp_global}.
	\end{proof}

We next design the control gain and the noise scales for a prescribed
infinite-horizon privacy budget. 
	
	\begin{theorem}
		\label{thm:inf_history_dp}
		Suppose \eqref{eq:S_exp_global} holds for every $k\ge0$ with
known constants $M\ge1$ and $\rho_H>0$. Such constants can be obtained
from Theorem~\ref{thm:history_sensitivity_global}, or from an
independently proved sharper sensitivity bound.
Let the control gain and the noise scales be designed as
		\begin{equation}\label{eq:param_design}
			c(k)=\frac{a_1}{(k+a_2)^\beta},
			\quad
			b_j(k)=\underline{b}(k+a_2)^\gamma,
		\end{equation}
		where $a_1,a_2,\underline{b}>0$ and $\gamma\ge0$.
		Then, for any predefined privacy level $\epsilon^\star>0$, the mechanism
		$\mathcal M^\infty$ achieves $\epsilon^\star$-differential privacy over
		the infinite time horizon if one of the following conditions holds:
		
		1) $\beta=1$, $\rho_H a_1+\gamma>1$, and
		\begin{align}\label{eq:eps_star_case1}
			\varepsilon(\infty)
			\le
			\sum_{k=0}^{d}
			\frac{\Delta_H}{\underline{b}(k+a_2)^\gamma}
			+
			\frac{
				M\Delta_H a_2^{1-\gamma}
			}{
				\underline{b}(\rho_H a_1+\gamma-1)
			}
			\le \epsilon^\star .
		\end{align}
		
		2) $0<\beta<1$ and
		\begin{equation}\label{eq:eps_star_case2}
		\begin{aligned}
			\varepsilon(\infty)
			&\le
			\sum_{k=0}^{d}
			\frac{\Delta_H}{\underline{b}(k+a_2)^\gamma}
			\\
			&\quad+
			\frac{
				M\Delta_H
				\exp\!\left(
				\frac{\rho_H a_1 a_2^{1-\beta}}{1-\beta}
				\right)
			}{
				\underline{b}(1-\beta)
			}
			\left(
			\frac{1-\beta}{\rho_H a_1}
			\right)^{\frac{1-\gamma}{1-\beta}}
			\\
			&\qquad\times\mathcal I\!\left(
			\frac{1-\gamma}{1-\beta},
			\frac{\rho_H a_1 a_2^{1-\beta}}{1-\beta}
			\right)
			\le \epsilon^\star ,
		\end{aligned}
		\end{equation}
where $\mathcal I(a,z):=\int_z^\infty u^{a-1}e^{-u}\,du$
is the upper incomplete gamma function. Here $z>0$, so this integral
is well defined for all real values of the displayed first argument.
	\end{theorem}
	
	\begin{proof}
	By Theorem~\ref{thm:history_dp_global}, for any finite time horizon $T$,
	the privacy loss satisfies
	$
	\varepsilon(T)
	=
	\sum_{k=0}^{T}\frac{S(k)}{b_{\min}(k)} .
	$
	Under \eqref{eq:param_design}, we have
	$
	b_{\min}(k)=\underline{b}(k+a_2)^\gamma .
	$
	If the infinite series is finite, then
	$
	\varepsilon(\infty)
	:=
	\lim_{T\to\infty}\varepsilon(T)
	=
	\sum_{k=0}^{\infty}
	\frac{S(k)}{b_{\min}(k)} .
	$
	Since $S(k)\le \Delta_H$ for $0\le k\le d$, we split the sum as
	\begin{equation}\label{eq:sum_split_dp}
		\varepsilon(\infty)
		\le
		\sum_{k=0}^{d}
		\frac{\Delta_H}{b_{\min}(k)}
		+
		\sum_{k=d+1}^{\infty}
		\frac{S(k)}{b_{\min}(k)} .
	\end{equation}
	For $k\ge d+1$, substituting \eqref{eq:param_design} into the sensitivity
	bound \eqref{eq:S_exp_global} yields
	\begin{align}\label{eq:sum_integral_bound}
		\sum_{k=d+1}^{\infty}
		\frac{S(k)}{b_{\min}(k)}
		\le\;
		\sum_{k=d+1}^{\infty}
		\frac{
			M\Delta_H
		}{
			\underline{b}(k+a_2)^\gamma
		}
			\exp\!\left(
		-\rho_H a_1
		\sum_{s=0}^{k-d-1}
		\frac{1}{(s+a_2)^\beta}
		\right).
	\end{align}
	Since $(s+a_2)^{-\beta}$ is nonincreasing, we have
	$
	\sum_{s=0}^{k-d-1}
	\frac{1}{(s+a_2)^\beta}$ $
	\ge
	\int_{0}^{k-d}
	\frac{1}{(x+a_2)^\beta}\,dx .
	$
	
	\textit{Case 1:} $\beta=1$.
Integration gives
\[
\int_0^{k-d}\frac{dx}{x+a_2}
=\log\!\left(\frac{k-d+a_2}{a_2}\right).
\]
Consequently,
\[
\exp\!\left(-\rho_Ha_1\int_0^{k-d}\frac{dx}{x+a_2}\right)
=
\left(\frac{a_2}{k-d+a_2}\right)^{\rho_Ha_1}.
\]

Noticing  $\gamma\ge0$, we have
	$
	(k+a_2)^{-\gamma}
	\le
	(k-d+a_2)^{-\gamma}.
	$
	Thus,
	\begin{align}
		\sum_{k=d+1}^{\infty}
		\frac{S(k)}{b_{\min}(k)}
		\le\;&
		\frac{M\Delta_H a_2^{\rho_H a_1}}{\underline{b}}
		\sum_{k=d+1}^{\infty}
		(k-d+a_2)^{-(\rho_H a_1+\gamma)}
		\notag\\
		\le\;&
		\frac{M\Delta_H a_2^{\rho_H a_1}}{\underline{b}}
		\int_{0}^{\infty}
		(x+a_2)^{-(\rho_H a_1+\gamma)}\,dx
			=\;
		\frac{
			M\Delta_H a_2^{1-\gamma}
		}{
			\underline{b}(\rho_H a_1+\gamma-1)
		},
	\end{align}
	where the last equality holds because $\rho_H a_1+\gamma>1$.
	Combining this estimate with \eqref{eq:sum_split_dp} and
	\eqref{eq:param_design} gives \eqref{eq:eps_star_case1}.
	
	\textit{Case 2:} $0<\beta<1$.
	In this case,
	\[
	\begin{aligned}
	\int_{0}^{k-d}
	\frac{1}{(x+a_2)^\beta}\,dx
	&=
	\frac{(k-d+a_2)^{1-\beta}-a_2^{1-\beta}}
	{1-\beta}.
	\end{aligned}
	\]
	we have
	\begin{align}
		\sum_{k=d+1}^{\infty}
		\frac{S(k)}{b_{\min}(k)}
			&\le
		\frac{M\Delta_H}{\underline{b}}
		\exp\!\left(
		\frac{\rho_H a_1 a_2^{1-\beta}}{1-\beta}
		\right)\notag\\
		&\quad\times\sum_{k=d+1}^{\infty}
		(k-d+a_2)^{-\gamma}
		\exp\!\left(
		-\frac{\rho_H a_1}{1-\beta}
		(k-d+a_2)^{1-\beta}
		\right).
	\end{align}
	Since $\gamma\ge0$, the function
	$
	x\mapsto
	(x+a_2)^{-\gamma}
	\exp\!\big(
	-\frac{\rho_H a_1}{1-\beta}
	(x+a_2)^{1-\beta}
	\big)
	$
	is nonincreasing on $[0,\infty)$. Hence,
	\begin{align}
		\sum_{k=d+1}^{\infty}
		\frac{S(k)}{b_{\min}(k)}
			&\le
		\frac{M\Delta_H}{\underline{b}}
		\exp\!\left(
		\frac{\rho_H a_1 a_2^{1-\beta}}{1-\beta}
		\right)\notag\\
		&\quad\times\int_{0}^{\infty}
		(x+a_2)^{-\gamma}
				\exp\!\left(
		-\frac{\rho_H a_1}{1-\beta}
		(x+a_2)^{1-\beta}
		\right)\,dx .
	\end{align}
	Let
	$
	y=
	\frac{\rho_H a_1}{1-\beta}
	(x+a_2)^{1-\beta}.
	$
	Then,
	$
	dx
	=
	\frac{1}{1-\beta}
	\left(
	\frac{1-\beta}{\rho_H a_1}
	\right)^{\frac{1}{1-\beta}}
	y^{\frac{\beta}{1-\beta}}\,dy .
	$
	It follows that
	\begin{align}
		\int_{0}^{\infty}
		(x+a_2)^{-\gamma}
		\exp\!\left(
		-\frac{\rho_H a_1}{1-\beta}
		(x+a_2)^{1-\beta}
		\right)\,dx
		&=
		\frac{1}{1-\beta}
		\left(
		\frac{1-\beta}{\rho_H a_1}
		\right)^{\frac{1-\gamma}{1-\beta}}\notag\\
		&\quad\times\mathcal I\!\left(
		\frac{1-\gamma}{1-\beta},
		\frac{\rho_H a_1 a_2^{1-\beta}}{1-\beta}
		\right).
	\end{align}
	Hence,
	\begin{align}
		\sum_{k=d+1}^{\infty}
		\frac{S(k)}{b_{\min}(k)}
		\le\;
		\frac{
			M\Delta_H
			\exp\!\left(
			\frac{\rho_H a_1 a_2^{1-\beta}}{1-\beta}
			\right)
		}{
			\underline{b}(1-\beta)
		}
		\left(
		\frac{1-\beta}{\rho_H a_1}
		\right)^{\frac{1-\gamma}{1-\beta}}
			\mathcal I\!\left(
		\frac{1-\gamma}{1-\beta},
		\frac{\rho_H a_1 a_2^{1-\beta}}{1-\beta}
		\right).
	\end{align}
	Combining this estimate with \eqref{eq:sum_split_dp} and
	\eqref{eq:param_design} gives \eqref{eq:eps_star_case2}.
	
	Consequently, if the corresponding upper bound in
	\eqref{eq:eps_star_case1} or \eqref{eq:eps_star_case2} is no larger than
	$\epsilon^\star$, then
	$
	\varepsilon(\infty)\le \epsilon^\star .
	$
For every finite $T$, the privacy inequality therefore holds
with the same constant $\epsilon^\star$. It follows for every
finite-prefix cylinder event in the infinite output space. Such
cylinders form an algebra generating the product Borel
$\sigma$-algebra. The inequality is preserved under increasing unions
and decreasing intersections by continuity of probability measures;
the monotone class theorem therefore extends it to every measurable
event. This proves the stated infinite-horizon privacy guarantee.
The right-hand side of either \eqref{eq:eps_star_case1} or
\eqref{eq:eps_star_case2} is $\mathcal B/\underline b$, with
$\mathcal B$ independent of $\underline b$. Consequently, computing
the sensitivity constants first and choosing
$\underline b\ge\mathcal B/\epsilon^\star$ gives a constructive design.
	\end{proof}

\begin{corollary}\label{cor:joint_privacy_consensus}
Suppose the communication graph contains a spanning tree and
$\kappa_i>0$ for every $i$. Choose the parameters in
\eqref{eq:param_design} with
$\beta\in(1/2,1]$ and $0\le\gamma<\beta-1/2$.
Let $M,\rho_H$ be certified constants in \eqref{eq:S_exp_global}.
If $\beta=1$, additionally require $\rho_Ha_1+\gamma>1$.
Select $\underline b$ using the corresponding bound in
Theorem~\ref{thm:inf_history_dp}.
Then the algorithm achieves mean square strong consensus, almost sure
strong consensus, and infinite-horizon $\epsilon^\star$-differential
privacy simultaneously.
\end{corollary}

\begin{proof}
Since $c(k)\to0$, the small-gain condition
\eqref{eq:cbar_A} holds eventually for every nonzero Laplacian
eigenvalue. Moreover, $\sum_kc(k)=\infty$ and
\[
\sum_{k=0}^\infty c(k)^2\bar b(k)^2
=
a_1^2\underline b^{\,2}
\sum_{k=0}^\infty(k+a_2)^{-2\beta+2\gamma}<\infty.
\]
Theorem~\ref{thm:dt_strong_consensus} gives both strong consensus
conclusions, and Theorem~\ref{thm:inf_history_dp} gives privacy.
\end{proof}

\section{Simulation}\label{sec:simulation}

We consider five scalar agents with delay $d=1$. The labeled graph in
Fig.~\ref{fig:sim_topology} has adjacency matrix
\[
A=\frac12
\begin{pmatrix}
0&0&0&1&1\\
1&0&0&0&1\\
1&1&0&0&0\\
0&1&1&0&0\\
0&0&1&1&0
\end{pmatrix}.
\]
Thus $K=I_5$, the graph is strongly connected and balanced, and
$\pi=\mathbf1_5/5$. The deterministic delayed history is
\[
\begin{aligned}
x(-1)&=(-1.20,1.80,2.00,-1.60,2.00)^\top,\\
x(0)&=(-2.00,1.40,2.10,-1.20,1.80)^\top.
\end{aligned}
\]
The public parameter sequences are
\[
c(k)=\frac{20}{k+100},\qquad
b_j(k)=0.27(k+100)^{0.05},\quad j=1,\dots,5.
\]
We take $\Delta_H=0.05$ and prescribe $\epsilon^\star=4$.

\begin{figure}[t]
\centering
\includegraphics[width=0.4\linewidth]{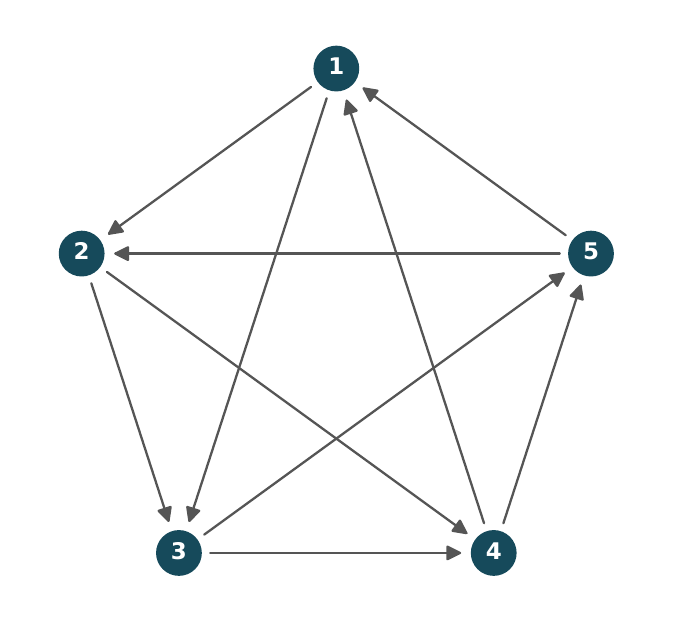}
\caption{Labeled directed communication graph. Every edge has weight $0.5$.}
\label{fig:sim_topology}
\end{figure}

The nonzero Laplacian eigenvalues are
$1.25\pm0.77\,\mathrm i$ and
$1.25\pm0.18\,\mathrm i$.
Since $\bar c=c(0)=0.2$, direct substitution gives
\[
\min_{\lambda_i\ne0}
\{2\Re\lambda_i-3|\lambda_i|^2\bar c\}
\simeq1.21>0.
\]
For the squared decay parameter $\eta_L=0.25$,
$\max_{\lambda_i\ne0}h_i(\eta_L)\simeq-0.23<0$.
Also $\alpha_H=2-3\bar c=1.4>0$ and
$2\beta-2\gamma=1.9>1$, so the strong-consensus conditions hold.

For this example, the privacy budget can be certified by a sharper
sensitivity estimate. For $g_t(u)=\Gamma(u,t)$ with $\lambda=1$, the ratio
$R_u=g_t(u)/g_t(u-1)$ satisfies
$R_{u+1}=1-c(u)/R_u$. Since $0<c(u)\le0.2$, the interval
$[(1+\sqrt{0.2})/2,1]$ is invariant. Consequently,
\[
g_t(u)\le
\exp\!\left(-\sum_{s=t+1}^{u-1}c(s)\right),\qquad u\ge t+1.
\]
The history identity \eqref{eq:history_impulse_identity} and
$c(0)e^{c(1)}<1$ then give, for $r=-1,0$ and $u\ge0$,
\[
|\Gamma_h(u,r)|
\le e^{0.2}
\exp\!\left(-\sum_{s=0}^{u-1}c(s)\right).
\]
Thus \eqref{eq:S_exp_global} holds with
$M=e^{0.2}$ and the amplitude rate $\rho_H=1$.
Theorem~\ref{thm:inf_history_dp} now yields
\begin{equation}\label{eq:simulation_certified_budget}
	\begin{aligned}
		\varepsilon(\infty)
		&\le
		\frac{0.05}{0.27}
		\left(100^{-0.05}+101^{-0.05}\right)\\
		&\quad+
		\frac{e^{0.2}\,0.05\,100^{0.95}}
		{0.27(20+0.05-1)}
		\simeq1.237249430<4.
	\end{aligned}
\end{equation}

\begin{figure}[t]
	\centering
	\includegraphics[width=0.8\linewidth]{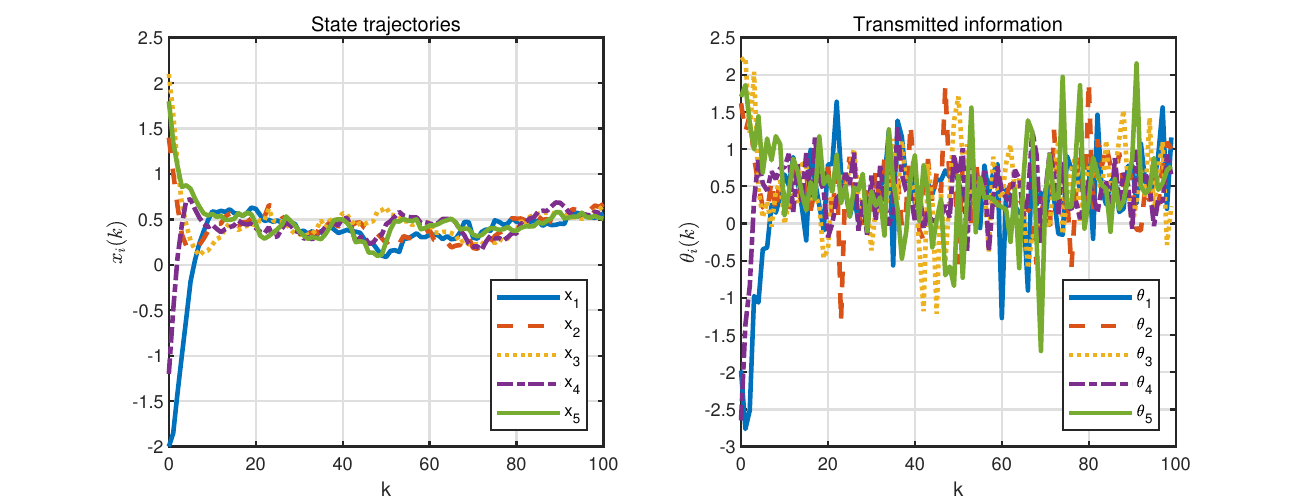}
	\caption{State trajectories and transmitted information under the proposed
		increasing privacy-noise mechanism.}
	\label{fig:state_theta}
\end{figure}

\begin{figure}[t]
	\centering
	\includegraphics[width=0.8\linewidth]{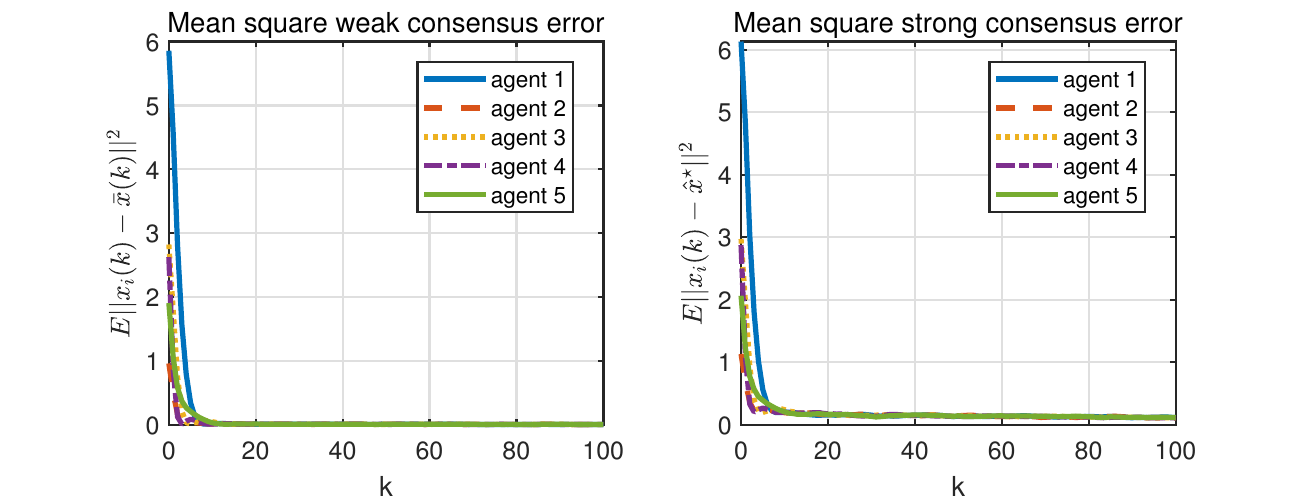}
	\caption{Mean square weak consensus errors and mean square strong consensus
		errors.}
	\label{fig:ms_weak_strong}
\end{figure}

\begin{figure}[t]
	\centering
	\includegraphics[width=0.5\linewidth]{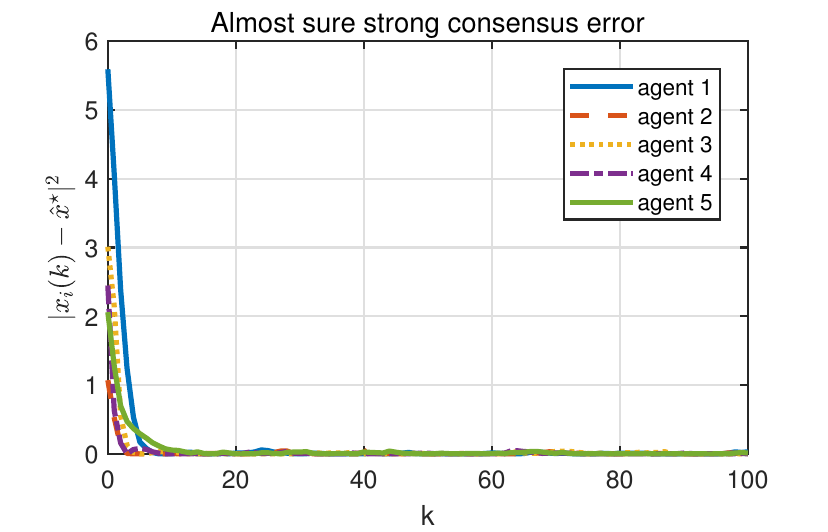}
	\caption{Almost sure consensus errors.}
	\label{fig:as}
\end{figure}

Figure~\ref{fig:state_theta} shows that the agent states approach one another
while the transmitted messages remain perturbed by the increasing Laplace
noise. Figures~\ref{fig:ms_weak_strong} and~\ref{fig:as} illustrate the
mean-square and sample-path consensus behavior, respectively. These finite
simulations are consistent with the theoretical results, while the privacy
level is certified analytically by
\eqref{eq:simulation_certified_budget}.

\FloatBarrier
\section{Conclusion}
	\label{sec:conclusion}
	
	This paper studied differentially private consensus for discrete-time
	MASs with communication delays. The delayed initial history
	was regarded as the private data, and a corresponding adjacency relation and
	privacy mechanism were developed. Under suitable gain-noise conditions, mean
	square consensus, almost sure consensus, and infinite-horizon
	$\epsilon^\star$-differential privacy were established. 
	
	Several problems remain open for future research. First, it would be
	interesting to extend the proposed framework to time-varying or randomly
	switching communication topologies. Second, event-triggered communication
	could be incorporated to reduce communication costs while preserving
	differential privacy for delayed histories. Finally, the present work focuses
	on first-order dynamics, and extending the analysis to higher-order or
	general linear MASs with delays is another meaningful
	direction. 

	\appendix
		\section{Proof of Lemma~\ref{lem:dt_fundamental_exp_A}}

We first treat $d=0$. In that case
$\Gamma(k,t)=\prod_{j=t}^{k-1}(1-\lambda c(j))$.
For $j\ge k_0$,
\[
|1-\lambda c(j)|^2
\le1-\alpha c(j)\le e^{-\alpha c(j)}.
\]
Condition \eqref{eq:h_varrho_condition} implies $\varrho<\alpha$.
An explicit constant for the finitely many exceptional factors is
\[
C=\max_{0\le t\le k\le k_0}
\left\{
\left|\prod_{j=t}^{k-1}(1-\lambda c(j))\right|^2
e^{\varrho\sum_{j=t}^{k-1}c(j)}
\right\}\ge1.
\]
Splitting a product at $k_0$ proves \eqref{eq:Gamma_exp_A}
for all $k\ge t$. Here $\Gamma_h(k,0)=\Gamma(k,0)$, so $C_h=C$.

Now let $d\ge1$. Fix $t\ge0$ and write $\gamma(k)=\Gamma(k,t)$.
Then $\gamma(t)=1$, $\gamma(j)=0$ for $j<t$, and
\begin{equation}\label{eq:gamma_rec_t}
\gamma(k+1)=\gamma(k)-\lambda c(k)\gamma(k-d),\qquad k\ge t.
\end{equation}
Introduce
\begin{equation}\label{eq:x_def_t}
q(k)=\gamma(k)-\lambda\sum_{j=k-d}^{k-1}c(j+d)\gamma(j).
\end{equation}
Shifting the summation window in \eqref{eq:x_def_t} gives
\begin{equation}\label{eq:x_rec_t}
q(k+1)=q(k)-\lambda c(k+d)\gamma(k).
\end{equation}
Define
\begin{equation}\label{eq:W_def_A_t}
W(k)=|\lambda|^2\sum_{j=k-d}^{k-1}
(j-k+d+1)c(j+d)^2|\gamma(j)|^2,\qquad
V_t(k)=|q(k)|^2+W(k).
\end{equation}
The subscript $t$ records the dependence of the solution on its
impulse time. The cross term in the difference of $|q(k)|^2$
satisfies
\[
\begin{aligned}
&2\Re\left[
\overline{\lambda\sum_{j=k-d}^{k-1}c(j+d)\gamma(j)}
\,\lambda c(k+d)\gamma(k)\right]\\
&\qquad\le
|\lambda|^2\sum_{j=k-d}^{k-1}c(j+d)^2|\gamma(j)|^2
+d|\lambda|^2c(k+d)^2|\gamma(k)|^2,
\end{aligned}
\]
by $2|a||b|\le d^{-1}|a|^2+d|b|^2$ and Cauchy--Schwarz.
Also,
\[
W(k+1)-W(k)
=
|\lambda|^2
\left[d\,c(k+d)^2|\gamma(k)|^2
-\sum_{j=k-d}^{k-1}c(j+d)^2|\gamma(j)|^2\right].
\]
Consequently, for $k\ge m(t):=\max\{t,k_0\}$,
\begin{equation}\label{eq:V_decay_A_t}
V_t(k+1)-V_t(k)
\le-\alpha c(k+d)|\gamma(k)|^2.
\end{equation}
From \eqref{eq:x_def_t}--\eqref{eq:x_rec_t} and
$|a+b|^2\le2|a|^2+2|b|^2$, one also obtains
\begin{equation}\label{eq:V_upper_for_swap_t}
V_t(k+1)
\le C_0|\gamma(k)|^2+
C_1\sum_{j=k-d}^{k-1}c(j+d)^2|\gamma(j)|^2,
\end{equation}
with exactly the constants $C_0,C_1$ in the lemma.

For $\rho>0$, put
\[
S_t(k)=\sum_{i=t}^{k-1}c(i+d),\qquad
\Phi_t(k)=e^{\rho S_t(k)}V_t(k),\qquad k\ge t.
\]
For $k<t$, use the convention
$S_t(k)=-\sum_{i=k}^{t-1}c(i+d)$ whenever needed below.
Using $e^u-1\le ue^u$, \eqref{eq:V_decay_A_t}, and
\eqref{eq:V_upper_for_swap_t}, we get, for $k\ge m(t)$,
\[
\begin{aligned}
\Phi_t(k+1)-\Phi_t(k)
&\le
a_\rho c(k+d)e^{\rho S_t(k)}|\gamma(k)|^2\\
&\quad+\rho e^{\rho\bar c}C_1 c(k+d)e^{\rho S_t(k)}
\sum_{j=k-d}^{k-1}c(j+d)^2|\gamma(j)|^2,
\end{aligned}
\]
where $a_\rho=\rho e^{\rho\bar c}C_0-\alpha$.
For $\ell>m=m(t)$, changing the order of the double sum gives
\[
\begin{aligned}
D_\ell
&:=\sum_{k=m}^{\ell-1}
c(k+d)e^{\rho S_t(k)}
\sum_{j=k-d}^{k-1}c(j+d)^2|\gamma(j)|^2\\
&=
\sum_{j=m-d}^{\ell-2}c(j+d)^2|\gamma(j)|^2
\sum_{k=\max\{m,j+1\}}^{\min\{\ell-1,j+d\}}
c(k+d)e^{\rho S_t(k)}.
\end{aligned}
\]
For every index in the inner sum,
$S_t(k)-S_t(j)\le d\bar c$, because all gain indices
$q+d$ in the intervening window are at least $m\ge k_0$.
There are at most $d$ terms. Since $c(j+d)\le\bar c$ for
$j\ge m-d$, it follows that
\[
D_\ell\le
d\bar c^2e^{\rho d\bar c}
\sum_{j=m-d}^{\ell-2}
c(j+d)e^{\rho S_t(j)}|\gamma(j)|^2.
\]
Write
\[
f_{t,j}:=c(j+d)e^{\rho S_t(j)}|\gamma(j)|^2,\qquad
B_\rho:=\rho C_1d\bar c^2e^{\rho(d+1)\bar c}.
\]
Summation of the one-step estimate therefore yields the bound
\[
\Phi_t(\ell)
\le\Phi_t(m)
+a_\rho\sum_{j=m}^{\ell-1}f_{t,j}
+B_\rho\sum_{j=m-d}^{\ell-2}f_{t,j}.
\]
Only the sum with the nonnegative coefficient $B_\rho$ is enlarged.
Separating its finite boundary part gives
\begin{equation}\label{eq:Phi_final_sum_t}
\Phi_t(\ell)
\le
\Phi_t(m)+B_\rho\sum_{j=m-d}^{m-1}f_{t,j}
+h(\rho)\sum_{j=m}^{\ell-1}f_{t,j}.
\end{equation}
Set $\rho=\varrho$. Since $h(\varrho)<0$, the last term can be
discarded. For $t\ge k_0$, the boundary terms vanish because
$\gamma(j)=0$ for $j<t$, while $\Phi_t(t)=V_t(t)=1$.
For $t<k_0$, both $\Phi_t(k_0)$ and the boundary term involve
only finite-prefix recursions.

Here are explicit choices making the uniformity of the constants clear.
With $\rho=\varrho$, let
\[
\begin{aligned}
\mathcal B_t
&:=\Phi_t(k_0)+B_{\varrho}
\sum_{j=k_0-d}^{k_0-1}f_{t,j},\qquad 0\le t<k_0,\\
M_V&:=\max\left\{
1,\ \max_{0\le t<k_0}\mathcal B_t,\
\max_{\substack{0\le t<k_0\\t\le k<k_0}}\Phi_t(k)
\right\}.
\end{aligned}
\]
Maxima over empty finite sets in these constant definitions are
understood as zero. The preceding estimate proves
$\Phi_t(k)\le M_V$ for all $k\ge t\ge0$. Moreover,
\[
|\gamma(k)|^2
\le2|q(k)|^2+
2d|\lambda|^2\sum_{j=k-d}^{k-1}c(j+d)^2|\gamma(j)|^2
\le2(d+1)V_t(k).
\]
Define the finite window bound
\[
A_d:=\max\left\{d\bar c,\
\max_{0\le t<k_0}\sum_{i=t}^{t+d-1}c(i)\right\}.
\]
The identity
\[
\sum_{i=t}^{k-1}c(i+d)
\ge\sum_{i=t}^{k-1}c(i)-\sum_{i=t}^{t+d-1}c(i)
\]
then gives \eqref{eq:Gamma_exp_A} with the computable constant
\begin{equation}\label{eq:constructive_C}
C:=2(d+1)M_Ve^{\varrho A_d}\ge1.
\end{equation}
All finite-prefix quantities above depend only on
$\lambda,d,\bar c,\varrho$ and $\{c(0),\dots,c(k_0+d)\}$.

Finally, the history and impulse solutions have the exact relation
\begin{equation}\label{eq:history_impulse_identity}
\begin{aligned}
\Gamma_h(k,0)&=\Gamma(k,0),\\
\Gamma_h(k,r)&=
\begin{cases}
0,&0\le k\le r+d,\\
-\lambda c(r+d)\Gamma(k,r+d+1),&k\ge r+d+1,
\end{cases}
\quad -d\le r<0 .
\end{aligned}
\end{equation}
Indeed, a negative-time history entry enters the future recursion
only once, as an impulse at time $r+d+1$. Applying the already
proved bound to this identity gives \eqref{eq:Gammah_exp_A} with the
constant $C_h$ specified in Lemma~\ref{lem:dt_fundamental_exp_A}.
The case $r=0$ is covered because $C_h\ge C$.
The proof is complete.\hfill\proofbox

\section{Proof of Lemma~\ref{lem:dt_kernel_toeplitz}}   
	By Lemma~\ref{lem:dt_fundamental_exp_A}, under \eqref{eq:cbar_A}, there exist constants $C>0$ and $\varrho>0$ such that
	\begin{equation}\label{eq:Gamma_exp_for_kernel}
		|\Gamma(k,s+1)|^2
		\le
		C\exp\!\Bigl(
		-\varrho\sum_{r=s+1}^{k-1}c(r)
		\Bigr),
		\qquad \forall\, k>s\ge 0.
	\end{equation}
	Hence, letting $\mu:=\varrho/2$, we obtain
	\begin{equation}\label{eq:Gamma_exp_sqrt_for_kernel}
		|\Gamma(k,s+1)|
		\le
		\sqrt{C}\exp\!\Bigl(
		-\mu\sum_{r=s+1}^{k-1}c(r)
		\Bigr),
		\qquad \forall\, k>s\ge 0.
	\end{equation}
	
	We first prove that $a_{k,s}\to0$ for each fixed $s$. Indeed, for fixed $s\ge0$,
	\[
	a_{k,s}
	=
	|\Gamma(k,s+1)|\,c(s)
	\le
	\sqrt{C}\,c(s)\,
	\exp\!\Bigl(
	-\mu\sum_{r=s+1}^{k-1}c(r)
	\Bigr).
	\]
	Since $\sum_{k=0}^{\infty} c(k) = \infty$ implies
	$\sum_{r=s+1}^{k-1}c(r)\to\infty$ as $k\to\infty$,
	it follows that $a_{k,s}\to0$ for each fixed $s\ge0$.
	
	We next prove that the row sums are uniformly bounded. By \eqref{eq:cbar_A}, there exists an integer $k_0\ge0$ such that
	$\bar c:=\sup_{k\ge k_0}c(k)<\infty$.
	For $k\ge\max\{1,k_0\}$, write
\[
\sum_{s=0}^{k-1}a_{k,s}
	=
	\sum_{s=0}^{k_0-1}a_{k,s}
	+
	\sum_{s=k_0}^{k-1}a_{k,s},
	\]
	where the first sum is empty if $k_0=0$.
	
	For the finite prefix, by \eqref{eq:Gamma_exp_sqrt_for_kernel},
	$\sum_{s=0}^{k_0-1}a_{k,s}
	\le
	\sqrt{C}\sum_{s=0}^{k_0-1}c(s)
	=:M_0<\infty$.
	
	For the tail part, define
	\[
	\begin{aligned}
	B_s&:=\exp\!\Bigl(-\mu\sum_{r=s+1}^{k-1}c(r)\Bigr),\\
	A_s&:=\exp\!\Bigl(-\mu\sum_{r=s}^{k-1}c(r)\Bigr)
	=B_s e^{-\mu c(s)}.
	\end{aligned}
	\]
	Then, for $s\ge k_0$,
	$B_s-A_s
	=
	B_s\bigl(1-e^{-\mu c(s)}\bigr)$.
	Using the elementary inequality
	$1-e^{-u}\ge u/(1+u)$, $u\ge0$,
	and $c(s)\le \bar c$, we get
	\[
	B_s-A_s
	\ge
	B_s\frac{\mu c(s)}{1+\mu c(s)}
	\ge
	B_s\frac{\mu c(s)}{1+\mu\bar c}.
	\]
	Therefore,
	$B_s c(s)\le \frac{1+\mu\bar c}{\mu}(B_s-A_s)$.
	Combining this with \eqref{eq:Gamma_exp_sqrt_for_kernel}, we obtain
	\[
	a_{k,s}
	\le
	\sqrt{C}\,B_s c(s)
	\le
	\sqrt{C}\frac{1+\mu\bar c}{\mu}(B_s-A_s),
	\qquad s\ge k_0.
	\]
	Summing from $s=k_0$ to $k-1$ yields
	\begin{align*}
		\sum_{s=k_0}^{k-1}a_{k,s}
		\le
		\sqrt{C}\frac{1+\mu\bar c}{\mu}
		\sum_{s=k_0}^{k-1}(B_s-A_s)=
		\sqrt{C}\frac{1+\mu\bar c}{\mu}
		\bigl(B_{k-1}-A_{k_0}\bigr) \le
		\sqrt{C}\frac{1+\mu\bar c}{\mu},
	\end{align*}
	since $B_{k-1}=1$ and $A_{k_0}\ge0$.
	
Let
$M_{\rm pre}:=\max_{1\le k<k_0}\sum_{s=0}^{k-1}a_{k,s}$,
with $M_{\rm pre}=0$ if the index set is empty.
Consequently,
\[
\sup_{k\ge1}\sum_{s=0}^{k-1}a_{k,s}
\le
\max\left\{M_{\rm pre},
M_0+\sqrt C\,\frac{1+\mu\bar c}{\mu}\right\}<\infty.
\]
	
	Finally, let $\{u(s)\}$ be any deterministic sequence such that $u(s)\to0$. Since every convergent sequence is bounded, there exists $M_u>0$ such that $|u(s)|\le M_u$ for all $s$. Then, we have
	\[
	\Bigl|
	\sum_{s=0}^{k-1}\Gamma(k,s+1)c(s)u(s)
	\Bigr|
	\le
	\sum_{s=0}^{k-1}a_{k,s}|u(s)|.
	\]
	Because $a_{k,s}\to0$ for each fixed $s$ and $\sup_k\sum_{s=0}^{k-1}a_{k,s}<\infty$, the Toeplitz Lemma (\cite{Knopp:51}) applies and gives
	$\sum_{s=0}^{k-1}\Gamma(k,s+1)c(s)u(s)\to0$.
	This completes the proof.\hfill\proofbox
	
\end{document}